\documentclass[dvipsnames,table]{article}

\usepackage{reportstyle}
\usepackage{amsmath,amssymb,mathtools,amsthm}
\usepackage{graphicx,placeins,xspace}
\usepackage[utf8]{inputenc}
\usepackage[T1]{fontenc}
\usepackage{url,booktabs,microtype,colortbl,tabularx,siunitx}
\usepackage[labelfont=bf]{caption}
\usepackage{enumitem,sectsty,forest,fancyhdr}
\usetikzlibrary{arrows.meta,calc,positioning}

\let\cite\citep
\renewcommand{\headrulewidth}{1pt}
\makeatletter
\def\headrule{{\if@fancyplain\let\headrulewidth\plainheadrulewidth\fi
  \hrule\@height\headrulewidth\@width\headwidth\vskip-\headrulewidth}}
\makeatother

\definecolor{HYDarkBlue}{HTML}{2155EA}
\definecolor{HYLightBlue}{HTML}{A8DFF6}
\usepackage[
  colorlinks,
  linkcolor=HYDarkBlue,
  anchorcolor=HYDarkBlue,
  citecolor=HYDarkBlue,
  urlcolor=HYDarkBlue
]{hyperref}
\usepackage{cleveref,threeparttable}
\sectionfont{\color{HYDarkBlue}\fontfamily{zi4}\selectfont}
\subsectionfont{\color{HYDarkBlue}\fontfamily{zi4}\selectfont}

\theoremstyle{plain}
\newtheorem{theorem}{Theorem}[section]
\newtheorem{lemma}[theorem]{Lemma}
\newcommand{\topk}{Top-$K$\xspace}

\title{Sample-Guided Exact Top-$K$ Selection for Long-Context Sparse Attention}
\author{%
  \textbf{\small Siran Liu, \quad Yang Xue, \quad Theo Tang, \quad Changxu Shao, \quad
  Qian Cheng, \quad Haimeng Ren, \quad Donghua Jiang}\\[0.12em]
  \textbf{\small Haipeng Ming, \quad Lehua Ding, \quad Zhonghan Lin, \quad Shengying Wei, \quad
  Wei Liu\thanks{Correspondence to: Wei Liu.}, \quad Kai Liu, \quad Jianchen Zhu}\\[0.45em]
  {\normalfont\small Tencent Inc.}
  \vspace{0em}
}

\newcommand{\githubicon}{%
  \raisebox{-1.5pt}{\includegraphics[height=1.05em]{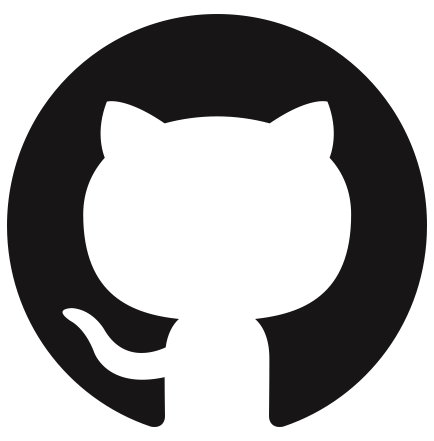}}}
\newcommand{\emailicon}{%
  \raisebox{-1.5pt}{\includegraphics[height=1.05em]{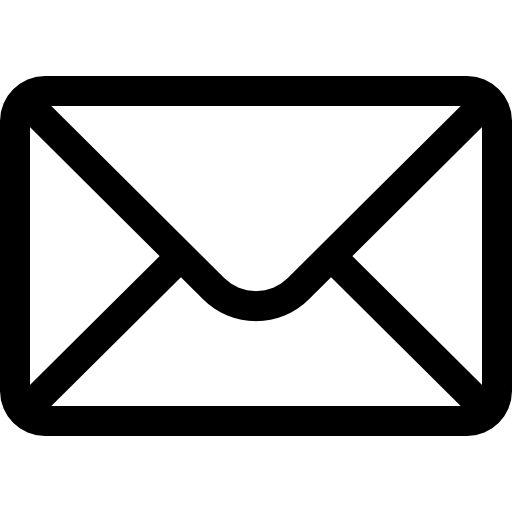}}}
\newcommand{\artifactlinks}{%
  \begin{tabular}{rl}
    \githubicon & \url{https://github.com/Tencent/hpc-ops} \\
    \emailicon & \href{mailto:sethranliu@tencent.com}{\nolinkurl{sethranliu@tencent.com}},
      \href{mailto:reedlauliu@tencent.com}{\nolinkurl{reedlauliu@tencent.com}}
  \end{tabular}%
}

\begin{document}

\maketitle
\thispagestyle{fancy}
\vspace{-7pt}

\begin{abstract}
Sparse attention bounds downstream attention work by retaining a fixed-size subset of indexed
tokens, but its standalone exact Top-$K$ stage must still process materialized score rows whose
length grows with context.  Production radix selectors discover their first actionable boundary
only after a complete-row pass, forcing another row-scale traversal before exact refinement.  We
observe that locating a compact upper tail requires substantially less resolution than identifying
the exact rank boundary, and that fixed-stride partial views of the current row remain calibrated
to the corresponding complete-row rank across ragged lengths.  We present \textbf{HPC-Ops Top-K},
a sample-guided exact selector for ragged sparse-attention score rows.  A fixed-stride view proposes a
row-local coarse boundary; the mandatory complete-row pass certifies its sufficiency, forms the
admitted candidate set, and initializes exact FP32 refinement over the unresolved frontier.  A
nested secondary boundary and exact recovery handle underfilled proposals before any output is
committed, so sampling controls common-path work but never correctness.  The GPU implementation
fuses complete-row certification and candidate formation, and combines persistent, KV-split, and
direct-exact execution behind graph-capturable ragged-row dispatch.  We evaluate HPC-Ops Top-K on
indexer scores from Hy4-Preview.  It outperforms the fastest verified external exact baseline by
$1.29$--$1.75\times$ across 20 operator configurations, with a $1.55\times$ geometric-mean speedup.
It further achieves $1.36\times$ and $1.48\times$ speedups on two framework-derived sparse-attention
traces.  The implementation is available in HPC-Ops, Tencent's open-source high-performance
operator library for LLM inference, at
\url{https://github.com/Tencent/hpc-ops}.
\end{abstract}

\section{Introduction}
\label{sec:intro}

Long-context inference increases both KV-cache traffic and the cost of determining which tokens
each query should attend to.  Sparse attention bounds the subsequent attention work by retaining
a fixed-size subset of the indexed context.  In indexer-based designs such as DeepSeek Sparse
Attention~\cite{deepseekv32}, an indexer materializes an FP32 relevance-score row for each
query, exact Top-$K$ identifies the selected positions, and the main attention accesses only the
corresponding keys and values.  Fixing $K$ bounds the downstream attention work, but it does not
bound selection: every score row must still be examined over its full valid length.  At hundreds
of thousands of indexed tokens and many ragged rows per invocation, exact Top-$K$ can therefore
remain on the critical path even after the surrounding kernels have been optimized.

The central cost is not the unavoidable inspection of an unsorted row, but when an exact
selector obtains a boundary that can reduce the active set.  Production radix selection first
streams the complete row into a digit histogram; only the completed histogram reveals the bucket
containing rank $K$.  Scores encountered before the histogram completes cannot be classified on
their first encounter, so candidate formation requires another row-scale traversal before exact
refinement. Other exact GPU selectors maintain ordered state, partition or reduce the row, or
predict boundaries from prior
invocations~\cite{shanbhag2018topk,air2023,radik2024,rtopk2025,gaihre2021drtopk,gvr2026}.  These
mechanisms trade row traffic for ordered work, resident state, data-dependent intermediate sets,
or reusable history.  Closing this gap requires a current-row boundary before the complete-row
pass, together with exact semantics, regular GPU execution, ragged-row support, CUDA-Graph
capture, and score traffic close to one complete-row scan.

Our key insight is that early reduction does not require the resolution of exact selection. Two
findings support this separation.  First, a large boundary error can leave only a compact upper
tail: distinguishing the value at rank $K$ from its immediate neighbor can demand an extremely
precise cutoff, whereas substantially coarser boundaries still leave compact upper tails across
the captured score rows.  Second, fixed-stride views track global upper-tail quantiles across
ragged lengths without relying on prior invocations.  Finite-population order statistics explain
this position-agnostic baseline, while the endpoint structure observed in sparse-attention
scores further reduces the estimation error.  Together, these findings make the current row
itself a low-cost source of an early boundary while leaving the unavoidable complete-row pass
authoritative for exactness.

We use this separation to build \textbf{HPC-Ops Top-K}, a three-phase sample-guided exact
selector. Phase~1 reads a fixed-stride view and proposes a row-local coarse boundary.  Phase~2
applies that boundary during the mandatory complete-row traversal, certifies the sufficiency of
the admitted candidate set, persists the admitted candidates, and constructs the leading
histogram for exact refinement on their first encounter.  If the proposal underfills, a nested
secondary boundary and, when necessary, exact coarse recovery establish a sufficient candidate
set before any output is committed.  Phase~3 performs exact FP32 radix refinement only over the
compact certified frontier, committing resolved digit groups and carrying only the unresolved
boundary frontier forward. Sampling therefore controls common-path work but never determines the
returned indices.  With sampling stride $s$, the proposal reads approximately $L/s$ scores
before the required complete-row traversal, rather than introducing another row-scale
boundary-localization pass.

The GPU implementation realizes the sampled pipeline within a shape-dependent portfolio of exact
kernels.  The fused complete-row pass uses vectorized score reads while performing boundary
classification, candidate collection, and construction of the leading exact histogram in one
traversal.  Subsequent refinement carries forward only the unresolved frontier rather than
revisiting resolved scores.  Persistent kernels keep each row and its compact state under one
CTA, KV-split kernels expose parallelism across a few long rows, and direct-exact kernels cover
complementary shapes for which sampling does not amortize.  These paths share graph-stable
ragged-row dispatch and the same exact FP32 refinement contract.

We evaluate HPC-Ops Top-K on indexer scores from Hy4-Preview~\cite{hy4preview2026} against
verified exact Top-$K$ implementations from vLLM, TensorRT-LLM, SGLang, FlashInfer, and
PyTorch~\cite{vllm2023,tensorrtllm2026,sglang2024,flashinfer2025,pytorch2019}.  HPC-Ops Top-K is
fastest in all 20 evaluated operator configurations, outperforming the fastest external baseline
by $1.29$--$1.75\times$ with a $1.55\times$ geometric-mean speedup.  On two framework-derived
sparse-attention traces, it achieves $1.36\times$ and $1.48\times$ speedups.  The implementation
is released as part of HPC-Ops, Tencent's open-source high-performance operator library for LLM
inference, at \url{https://github.com/Tencent/hpc-ops}.

Our contributions are:
\begin{itemize}[leftmargin=1.5em,itemsep=2pt,topsep=3pt]
  \item \textbf{Workload insight.}  We identify and analyze the separation between locating a
    compact upper tail and resolving its exact rank boundary, and show that fixed-stride views of
    the current row can locate a useful coarse boundary across ragged lengths.
  \item \textbf{Sample-guided exact selection.}  We develop a three-phase
    proposal--certification--refinement algorithm in which sampling reduces common-path work while
    complete-row certification, bounded recovery, and exact FP32 refinement preserve exact
    Top-$K$ semantics.
  \item \textbf{GPU realization and validation.}  We realize the design through fused
    complete-row processing and a graph-stable dispatcher over persistent, KV-split, and
    direct-exact kernels, bringing the principal sampled path close to the one-scan traffic floor.
    We validate the resulting operator at both operator and framework levels and release it as part
    of the open-source HPC-Ops library.
\end{itemize}

\clearpage
\section{Background}
\label{sec:bg}

\subsection{Cost of Standalone Top-K Selection}
\label{sec:bg:interface}

Sparse-attention models decide which key--value entries each query may attend to by ranking a
score vector.
In indexer-based designs such as DeepSeek Sparse Attention~\cite{deepseekv32},
GLM-5.3~\cite{glm53model2026}, and Hy4-Preview~\cite{hy4preview2026}, a lightweight indexer
scores every historical token against the current query,
\begin{equation}
  \boldsymbol{\ell}_{t} \;=\; \sum_{j=1}^{h} w_j \cdot
  \operatorname{ReLU}\!\left(\mathbf{q}_{t,j}\,\mathbf{K}^{\top}\right)
  \;,
  \label{eq:indexer}
\end{equation}
and the selector keeps the $K$ positions with the largest scores.
The downstream sparse attention consumes those indices, not the scores.

Two properties of this arrangement set up the rest of the report. The score matrix is
\emph{materialized} before selection, so the selector is a standalone operator over an existing
array. And its cost scales with context length while the attention behind it does not.

\paragraph{Operator contract.}
The operator maps a materialized score matrix and a per-row valid length to $K$ indices per row:
\[
  \underbrace{\mathbf{X} \in \mathbb{R}^{M \times N}}_{\texttt{float32}},\quad
  \underbrace{\mathbf{e} \in \mathbb{N}^{M}}_{\text{valid lengths}}
  \;\longmapsto\;
  \underbrace{\mathbf{I} \in \mathbb{Z}^{M \times K}}_{\texttt{int32}} .
\]
Here $M$ is the number of score rows and $N$ is their allocated, possibly
padded width.  Row $r$ has valid length $L_r:=e_r\le N$, valid positions
$[L_r]:=\{0,\ldots,L_r-1\}$, and scores $x_{r,i}:=\mathbf X_{r,i}$ for
$i\in[L_r]$.  For a single-row argument, we omit $r$ and write $L$, $[L]$, and
$x_i$.
When $L_r>K$, the operator returns $K$ distinct indices from
$[L_r]$.
For $L_r\le K$, row $r$ instead contains every valid index, and any remaining slots
$j\in\{L_r,\ldots,K-1\}$ are filled with the sentinel $I_{r,j}=-1$; the downstream
sparse-attention operator skips those entries.
Selection is \emph{exact} when no unselected valid element of a row exceeds any selected one.
Ties at rank $K$ leave the index set underdetermined, so the contract fixes the selected values
rather than a particular index permutation.
Inputs are \texttt{float32} and free of \texttt{NaN}; ${\pm}\infty$ and subnormals are ordered
numerically.

\paragraph{Deployment requirements.}
Four properties of the serving path decide whether a selector can be dropped into it, and they
later exclude designs that are otherwise attractive. Rows are \emph{ragged}: under causal
masking within a prefill chunk, $L_r$ varies widely, so a schedule that assigns equal work per
row is unbalanced. Rows are stored at a padded pitch whose alignment, not the logical length,
determines the vector width a kernel may use. The operator must be capturable in a CUDA graph,
which forbids reading a runtime statistic back to the host to choose a code
path~\cite{nvidia2026cuda}, and its workspace must be bounded independently of the batch. It
must accept scores from any producer, and it must be replayable offline on a saved score matrix
so that numerical regressions in a deep model remain tractable.

\paragraph{Why long rows are expensive.}
Once the indexer has materialized a score row, the selector must cover its $L$ valid entries,
whereas downstream sparse attention consumes only the $K$ selected positions.  We compare the
selector's row-scaled score traffic with the selected-position payload of sparse attention.  Let
$b_s$ be the bytes per materialized score, and let $b_a$ be the logical bytes consumed by sparse
attention per selected position, including its index and model-specific cached attention state.
For the selector, let $W_{\mathrm{score}}$ count logical global-memory score reads over the valid
row: complete traversals and any row fraction revisited
before reduction reaches a compact candidate frontier.  The resulting row-equivalent traffic is
$R_{\mathrm{eq}}:=W_{\mathrm{score}}/L$, while subsequent exact refinement contributes
candidate-scale work.  The corresponding
score-read traffic and selected-position payload are then
\begin{equation}
  \mathcal B_{\mathrm{scan}}=R_{\mathrm{eq}}L b_s,
  \qquad
  \mathcal B_{\mathrm{sparse}}=K b_a.
  \label{eq:traffic}
\end{equation}
Their ratio exposes the relevant scaling law:
\begin{equation}
  \frac{\mathcal B_{\mathrm{scan}}}{\mathcal B_{\mathrm{sparse}}}
  =\frac{R_{\mathrm{eq}}b_s}{b_a}\frac{L}{K}.
  \label{eq:ratio}
\end{equation}
For fixed $K$ and data representations, this ratio grows linearly with $L$; each additional
complete-row traversal increases it in the same proportion.

The standalone exact-selection contract also imposes a traffic floor.  An arbitrary unsorted row
cannot be certified without inspecting every valid score, so
\begin{equation}
  R_{\mathrm{eq}}\geq1
  \qquad\text{for any exact selector.}
  \label{eq:floor}
\end{equation}
Each complete global traversal of the valid row contributes $1$, while revisiting a fraction
of the row contributes the same fraction.  Values reused from registers or shared memory add no
global score read.  Once the row has been reduced to a compact frontier, its exact work is accounted
for separately.  Thus $R_{\mathrm{eq}}=1$ denotes one complete scan without another row-scaled
revisit.
For fixed $K$, $b_s$, and $b_a$, $R_{\mathrm{eq}}$ is the selector-controlled factor in
\cref{eq:ratio}, which makes it the traffic axis used to compare the mechanisms in
\cref{sec:bg:landscape}.

\subsection{Approaches to Exact Selection}
\label{sec:bg:landscape}

The textbook task is to return the $K$ largest entries of a row; GPU selectors differ in how
they obtain the boundary needed for that result.  \Cref{fig:taxonomy} groups exact mechanisms by
whether they maintain ordered state, discover a boundary through current-row partitioning,
reduce the row before an exact backend, or predict a boundary from prior invocations. A separate
body of work reduces the same system cost by moving selection into score production, sharing
indexing work across queries or layers, or relocating sparse-attention data and execution;
\cref{sec:related} covers it.

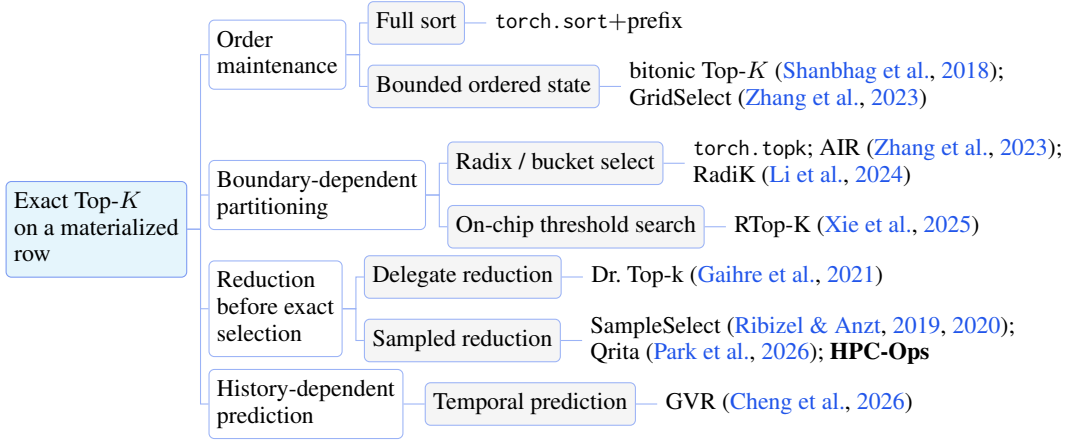
\begin{figure*}[!htbp]
  \centering
  \footnotesize
  \begin{forest}
    for tree={
      grow'=east,
      growth parent anchor=east,
      parent anchor=east,
      child anchor=west,
      anchor=west,
      draw=HYDarkBlue!45,
      rounded corners=2pt,
      align=left,
      inner sep=3pt,
      s sep=2.5pt,
      l sep=8pt,
      edge={HYDarkBlue!45},
      edge path={
        \noexpand\path[\forestoption{edge}]
          (!u.parent anchor) -- +(5pt,0) |- (.child anchor)\forestoption{edge label};
      },
    }
    [Exact \topk\\on a materialized\\row, fill=HYLightBlue!30, draw=HYDarkBlue!60
      [Order\\maintenance
        [{Full sort}, fill=gray!8
          [{\texttt{torch.sort}$+$prefix}, draw=none]
        ]
        [{Bounded ordered state}, fill=gray!8
          [{bitonic \topk~\cite{shanbhag2018topk};\\GridSelect~\cite{air2023}}, draw=none]
        ]
      ]
      [Boundary-dependent\\partitioning
        [{Radix / bucket select}, fill=gray!8
          [{\texttt{torch.topk}; AIR~\cite{air2023};\\RadiK~\cite{radik2024}}, draw=none]
        ]
        [{On-chip threshold search}, fill=gray!8
          [{RTop-K~\cite{rtopk2025}}, draw=none]
        ]
      ]
      [Reduction\\before exact\\selection
        [{Delegate reduction}, fill=gray!8
          [{Dr.\ Top-k~\cite{gaihre2021drtopk}}, draw=none]
        ]
        [{Sampled reduction}, fill=gray!8
          [{Sample\-Select~\cite{ribizel2019sampleselect,ribizel2020parallel};\\Qrita~\cite{qrita2026};
            \textbf{HPC-Ops}}, draw=none]
        ]
      ]
      [History-dependent\\prediction
        [{Temporal prediction}, fill=gray!8
          [{GVR~\cite{gvr2026}}, draw=none]
        ]
      ]
    ]
  \end{forest}
  \caption{Mechanisms for exact \topk selection over a materialized score row.
    The leaves name representative GPU realisations.
    The four branches organize the survey; \cref{tab:survey} then distills their seven mechanisms
    into deployment-facing cost and execution categories.
    HPC-Ops belongs to sampled reduction; methods that alter the score
    producer or the sparsity
    contract are covered separately in \cref{sec:related}.}
  \label{fig:taxonomy}
\end{figure*}

For deployment, however, the algorithmic labels are not enough.  The decisive question is what
information each mechanism makes available before or during the mandatory inspection of a long
row, and what resource or signal it exchanges for lower complete-row traffic.  We examine the
mechanisms in turn and then use \cref{tab:survey} to synthesize their representative work, score
traffic, GPU execution, and sensitivity to the input data.

\paragraph{Order-maintaining selection.}
Sorting the row and keeping a prefix costs $\log L$ traversals, and none of that work reflects
that only $K$ elements are wanted. Bitonic \topk removes most of it by discarding the losing
half after each merge stage~\cite{shanbhag2018topk}, and GridSelect maintains a grid-wide
bounded queue whose insertion position is computed with a warp ballot, so the expensive sort and
merge run only when the queue fills~\cite{air2023}. Both attain $R_{\mathrm{eq}}=1$, which is
the floor of \cref{eq:floor}.

What they pay for it is a footprint. The retained structure lives in shared memory and
registers, which caps the rank at $256$ for a bitonic network and $2048$ for a warp- or
grid-level queue~\cite{air2023}. The $\mathcal{O}(\log^2K)$ cost of the merge network compounds
this: by the same authors' measurement a queue is faster only below $K\approx256$, above which
radix partitioning wins on identical hardware~\cite{air2023}. Reaching the one-scan traffic
floor through bounded ordered state therefore does not remove its rank-scaling limitation.

\paragraph{Boundary-dependent partitioning.}
Radix and bucket selection build a histogram over successive digits of an order-preserving key,
locate the bin holding the $K$th element, and refine within it, giving $D_{\mathrm{rad}}$
iterations for a key processed by a $D_{\mathrm{rad}}$-digit schedule. AIR and RadiK are
representative GPU realisations~\cite{air2023,radik2024}; the CUDA path of \texttt{torch.topk}
uses the same radix-selection family~\cite{pytorchradix2026}.

Radix exposes regular SIMD work, and its later active ranges can shrink sharply once a boundary
bin is known.  Its initial cost, however, is fixed by boundary timing: a conventional path
spends one complete traversal locating the first boundary and another classifying the row
against it, so $R_{\mathrm{eq}}\geq2$ before later active-range work is considered.
\Cref{sec:bg:cuda-path} traces the CUDA dependency that makes this reread structural and
explains which later passes round fusion can overlap.

Its dependability is the other side of the same coin: no input makes it wrong, and skew costs
traffic rather than correctness. Adversarially clustered keys that share leading bits eliminate
nothing in the early iterations, and the mitigation chooses between re-reading the row and
spilling a candidate buffer according to which moves fewer bytes~\cite{air2023,radik2024}.  This
bounds the cost of skew rather than removing it.

A numeric threshold search reaches the same result differently, testing candidate thresholds and
counting qualifying elements with warp ballots~\cite{rtopk2025}. Loading the row into shared
memory once and iterating on chip attains $R_{\mathrm{eq}}=1$, and the published evaluation
covers rows of a few hundred to a few thousand elements. That prerequisite is the constraint
rather than the iteration count: the evaluated rows fit in CTA-local shared memory, whereas the
long sparse-attention rows considered here do not.  Once the row spills, each of the
$\mathcal{O}(\log L)$ search rounds becomes a global traversal.

\paragraph{Reduction before exact selection.}
Spending a little work to shrink what an exact backend must examine can keep complete-row
traffic near the one-scan floor without requiring the row to remain on chip.  Unlike the
preceding mechanisms, its excess work is governed by the size of an intermediate reduction
rather than by a structurally required full-row reread. Dr.\ Top-k partitions the row and uses
per-partition delegates to prove that whole partitions cannot
contribute~\cite{gaihre2021drtopk}. The bound never fails, but the subrange width that controls
it has an interior optimum: a wider subrange shrinks the delegate set and enlarges the second
stage that must then be concatenated. The saving therefore rests on a constant tuned to the data
rather than on a property of the row.

Sampling obtains its reduction signal from the current row.  SampleSelect estimates splitters
from a sample and partitions the full input
exactly~\cite{ribizel2019sampleselect,ribizel2020parallel}, while Qrita fits a sampled tail
model before narrowing the quantile with a multi-pivot search~\cite{qrita2026}.  Prof-K uses a
random sample to size a one-pass filter and recovers the exact Top-$K$ with user-specified high
probability~\cite{profk2026}.  SampleSelect and Qrita retain deterministic exact completion
through full-row verification; Prof-K instead adopts a probabilistic correctness contract.
Exactness alone, however, does not determine whether the reduction repays its sampling work.  On
long, ragged sparse-attention rows, the decisive quantity is the candidate volume left by a
sublinear per-row view: a useful estimate must reduce the exact backend to a compact,
$O(K)$-scale tail rather than another row-scale problem.

\paragraph{History-dependent prediction.}
Consecutive decode steps produce strongly correlated scores, so the previous step's selection is
an excellent initial threshold. GVR uses that guess to initialize a full-row secant threshold
search and then refines exactly. Across its reported DSA decode layers, Phase~2 averages
$1.1$--$2.7$ threshold-search iterations~\cite{gvr2026}. The guess is strong enough that Phase~2
approaches a single threshold-search traversal on high-correlation layers, making the approach
highly effective on decode-shaped work.

Its prerequisite, though, is a property of the schedule rather than of the data, and cannot be
supplied when absent. There is no previous step at the first chunk of a prefill, on a single
offline call, or when replaying a saved score matrix, and on those calls the approach falls back
to the partitioning path and inherits its multi-pass traffic. Full-row counting remains
unavoidable: the prior initializes the search but each threshold is checked against the current
row, and candidate formation adds another traversal.  With $n_{\mathrm{search}}$
threshold-search traversals, its row-scaled traffic therefore satisfies $R_{\mathrm{eq}}\geq
n_{\mathrm{search}}+1$ even when the initial proposal is strong.

\Cref{tab:survey} compares the mechanisms along four deployment-relevant dimensions. The
GPU-mapping column identifies the dominant execution structure rather than measured speed. Data
sensitivity describes how much the executed work changes with score distribution or proposal
quality at fixed $L$ and $K$: \emph{low} denotes a largely fixed schedule, \emph{moderate}
data-dependent candidate or refinement work, and \emph{high} a fast path governed by the initial
boundary estimate.  Because every mechanism retains exact completion, this grade concerns
performance rather than correctness.

\begin{table}[!htbp]
  \centering
  \begin{threeparttable}
  \caption{Deployment-facing comparison of the mechanisms in \cref{fig:taxonomy} over one valid
    row.  Work is representative rather than an implementation-level bound;
    $R_{\mathrm{eq}}$ reports row-scaled global score traffic, normalized by one complete
    valid-row read; compact candidate-only exact work remains explicit in the Work column.  All rows
    denote exact-completion mechanisms.}
  \label{tab:survey}
  \footnotesize
  \setlength{\tabcolsep}{3.4pt}
  \renewcommand{\arraystretch}{1.16}
  \begin{tabularx}{\linewidth}{@{}
    >{\raggedright\arraybackslash}p{.225\linewidth}
    >{\raggedright\arraybackslash}p{.25\linewidth}
    >{\centering\arraybackslash}p{.12\linewidth}
    >{\raggedright\arraybackslash}p{.18\linewidth}
    >{\raggedright\arraybackslash}X@{}}
    \toprule
    \textbf{Mechanism} & \textbf{Representative row work} &
      $\boldsymbol{R_{\mathrm{eq}}}$ & \textbf{GPU mapping} &
      \textbf{Data sensitivity} \\
    \midrule
    Full ordering & $\mathcal O(L\log L)$ & $\log L$
      & General-purpose &
      Low \\
    Bounded ordered state & $\mathcal O(L\log^2 K)$ & $1$
      & Rank-limited & Low \\
    \rowcolor{gray!7}
    \textbf{Radix / bucket select} & $\mathcal O(D_{\mathrm{rad}}L)$
      & $\geq 2$
      & \textbf{Strong SIMD} & Low--Moderate \\
    On-chip threshold search & $\mathcal O(n_{\mathrm{thr}}L)$ & $1$
      & Capacity-limited & Moderate \\
    Delegate reduction & $\mathcal O(L+D_{\mathrm{del}}\sqrt{LK})$
      & $1+f_{\mathrm{del}}$ & Two-stage & Moderate \\
    \rowcolor{HYLightBlue!18}
    \textbf{Sampled reduction} &
      $\mathcal O(L_{\mathrm{samp}}+L+W_{\mathrm{cand}})$
      & $1+L_{\mathrm{samp}}/L$
      & \textbf{Strong SIMD} & Moderate \\
    Temporal prediction &
      $\mathcal O((n_{\mathrm{search}}{+}1)L+W_{\mathrm{cand}})$
      & $n_{\mathrm{search}}+1$
      & Strong SIMD & High \\
    \bottomrule
  \end{tabularx}
  \begin{tablenotes}[flushleft]
    \footnotesize
    \item $D_{\mathrm{rad}}$, $D_{\mathrm{del}}$, $n_{\mathrm{thr}}$, and $n_{\mathrm{search}}$ count radix,
      delegate-refinement, on-chip threshold-search, and temporal threshold-search rounds; $f_{\mathrm{del}}$ is the row
      fraction revisited after delegate filtering; $L_{\mathrm{samp}}$ is the aggregate sample size;
      and $W_{\mathrm{cand}}$ is candidate-only exact work after row-scale reduction.
  \end{tablenotes}
  \end{threeparttable}
\end{table}
\FloatBarrier

\paragraph{What the landscape leaves open.}
Taken together, the survey and \cref{tab:survey} expose a common tradeoff.  Order-maintaining
and on-chip methods exchange score traffic for excess ordering work or resident state; reduction
methods exchange it for a data-dependent intermediate set; temporal prediction exchanges it for
prior invocation state.  Production radix avoids these capacity and history prerequisites and
retains regular GPU execution, but pays an initial full-row reread.

The missing combination is therefore long-row support, regular GPU execution, current-row state,
and score traffic near one valid-row scan.  Sampled reduction has the appropriate information
source, but the retained-tail size of a sublinear view at the target rank is the unresolved
property.  The next subsection isolates why production radix obtains its first boundary too late
to close this gap; \cref{sec:insight} then evaluates whether current-row views supply the
required compact tail.

\subsection{The Production Radix Bottleneck}
\label{sec:bg:cuda-path}

Production radix provides self-contained, regular GPU execution, but its gap from the one-scan
floor is set by when the rank boundary becomes available.  Its pass count is not merely the
number of digits in the key; it follows from a CUDA data dependency between histogram
construction and boundary-dependent classification. \Cref{fig:bg:radix-path} separates this
structural dependency from the execution pressure within each pass.

\paragraph{Structural traffic.}
The first pass contributes every key to an on-chip digit histogram, and only the completed bin
counts allow a prefix scan to expose the bucket containing rank $K$. The scores streamed through
the CTA before that point cannot yet be classified by the resulting boundary. The implementation
must therefore touch the row again or retain a row-sized copy. Round fusion overlaps each
boundary classification with the next digit histogram, reducing the worst-case score reads of a
four-digit schedule from $8L$ to $5L$~\cite{air2023}.  On well-spread keys, an 11-bit digit also
shrinks later active ranges far below $L$.  Neither effect can fuse classification into the
initial pass that creates the first boundary, so its histogram--classification pair still
contributes $2L$ reads before candidate shrinkage can help.

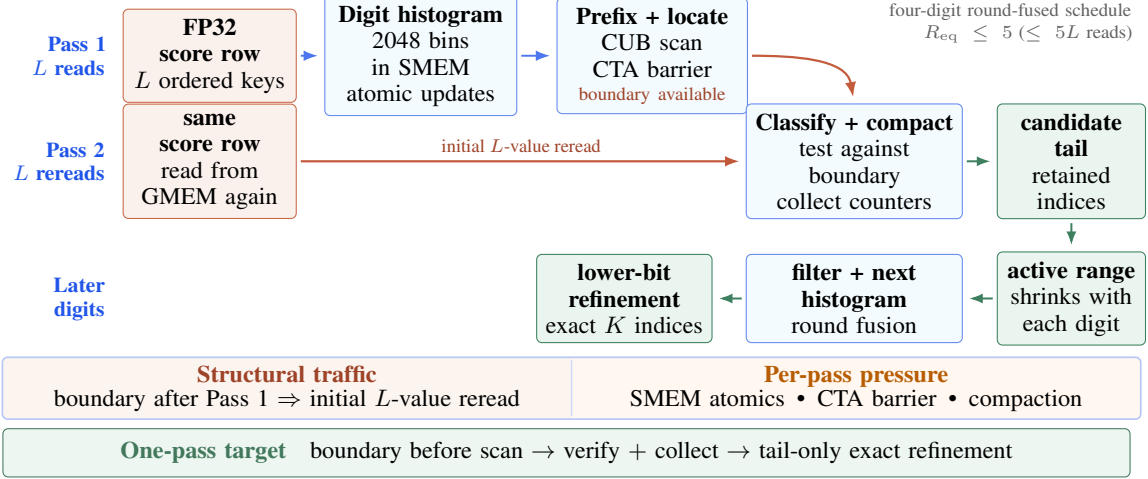
\begin{figure}[!htbp]
  \centering
  \resizebox{\linewidth}{!}{%
  \begin{tikzpicture}[
    font=\footnotesize,
    box/.style={
      draw=HYDarkBlue!70,
      fill=HYLightBlue!14,
      rounded corners=2pt,
      minimum height=1.13cm,
      align=center,
      inner sep=3.5pt
    },
    input/.style={
      box,
      draw=BrickRed!72,
      fill=BrickRed!5
    },
    tail/.style={
      box,
      draw=ForestGreen!68!black,
      fill=ForestGreen!7
    },
    flow/.style={-{Latex[length=1.9mm]},line width=0.72pt,draw=HYDarkBlue!80,
      shorten <=1.5pt,shorten >=2.4pt},
    reread/.style={-{Latex[length=1.9mm]},line width=0.82pt,draw=BrickRed!80,
      shorten <=1.5pt,shorten >=2.4pt},
    candidate/.style={-{Latex[length=1.8mm]},line width=0.75pt,
      draw=ForestGreen!65!black,shorten <=1.5pt,shorten >=2.4pt},
    dependency/.style={-{Latex[length=2.0mm]},line width=0.84pt,
      draw=BrickRed!78,shorten <=1.5pt,shorten >=2.4pt},
    lane/.style={font=\fontsize{8}{9}\selectfont\bfseries,align=right,text=black!72}
  ]
    \node[font=\normalsize\bfseries,text=HYDarkBlue,anchor=west] at (0,5.05)
      {Conventional production radix};
    \node[font=\scriptsize,text=black!62,text width=3.35cm,
      anchor=east,align=right] at (15.10,4.55)
      {four-digit round-fused schedule\\
       $R_{\mathrm{eq}}\leq5$ ($\leq5L$ reads)};

    \node[lane,text=HYDarkBlue,text width=1.35cm,anchor=east] at (1.48,4.12)
      {Pass 1\\$L$ reads};
    \node[input, text width=2.05cm] (row1) at (2.75,4.12)
      {\textbf{FP32 score row}\\$L$ ordered keys};
    \node[box, text width=2.30cm] (hist) at (5.55,4.12)
      {\textbf{Digit histogram}\\2048 bins in SMEM\\atomic updates};
    \node[box,text width=2.28cm,minimum height=1.34cm] (locate) at (8.62,4.12)
      {\textbf{Prefix + locate}\\CUB scan\\CTA barrier\\[-1pt]
       \textcolor{BrickRed!82!black}{\scriptsize boundary available}};

    \draw[flow] (row1.east) -- (hist.west);
    \draw[flow] (hist.east) -- (locate.west);

    \node[lane,text=HYDarkBlue,text width=1.35cm,anchor=east] at (1.48,2.72)
      {Pass 2\\$L$ rereads};
    \node[input,text width=2.05cm] (row2) at (2.75,2.72)
      {\textbf{same score row}\\read from GMEM again};
    \node[box,text width=2.62cm] (classify) at (11.30,2.72)
      {\textbf{Classify + compact}\\
       test against boundary\\collect counters};
    \node[tail,text width=1.72cm] (candidates) at (14.18,2.72)
      {\textbf{candidate tail}\\retained indices};

    \draw[reread] (row2.east) -- node[midway,above=3pt,font=\scriptsize,
      text=BrickRed!82!black,inner sep=0pt] {initial $L$-value reread}
      (classify.west);
    \draw[dependency] (locate.east) to[out=0,in=100] (classify.north);
    \draw[candidate] (classify.east) -- (candidates.west);

    \node[lane,text=HYDarkBlue,text width=1.35cm,anchor=east]
      at (1.48,0.90)
      {Later\\digits};
    \node[tail,text width=1.72cm] (active) at (14.18,0.90)
      {\textbf{active range}\\shrinks with each digit};
    \node[box,text width=2.62cm] (fusion) at (11.30,0.90)
      {\textbf{filter + next histogram}\\round fusion};
    \node[tail,text width=2.08cm] (output) at (8.25,0.90)
      {\textbf{lower-bit refinement}\\exact $K$ indices};

    \draw[candidate] (candidates.south) -- (active.north);
    \draw[candidate] (active.west) -- (fusion.east);
    \draw[candidate] (fusion.west) -- (output.east);

    \begin{scope}
      \clip[rounded corners=2pt] (0,0.12) rectangle (15.10,-0.70);
      \fill[BrickRed!4] (0,0.12) rectangle (7.55,-0.70);
      \fill[orange!6] (7.55,0.12) rectangle (15.10,-0.70);
    \end{scope}
    \draw[HYDarkBlue!32,line width=0.65pt,rounded corners=2pt]
      (0,0.12) rectangle (15.10,-0.70);
    \draw[HYDarkBlue!22,line width=0.55pt] (7.55,0.08) -- (7.55,-0.66);

    \node[font=\footnotesize,align=center,text width=7.00cm] at (3.775,-0.29) {
      \textcolor{BrickRed!82!black}{\textbf{Structural traffic}}\\[-1pt]
      boundary after Pass 1 $\Rightarrow$ initial $L$-value reread};
    \node[font=\footnotesize,align=center,text width=7.00cm] at (11.325,-0.29) {
      \textcolor{orange!72!black}{\textbf{Per-pass pressure}}\\[-1pt]
      SMEM atomics \,\textbullet\, CTA barrier \,\textbullet\, compaction};

    \node[draw=ForestGreen!58!black,fill=ForestGreen!6,rounded corners=2pt,
      minimum height=0.64cm,text width=14.74cm,align=center,inner sep=3pt,
      font=\footnotesize,anchor=north west] at (0,-0.82) {
      \textcolor{ForestGreen!58!black}{\textbf{One-pass target}}\quad
      boundary before scan $\rightarrow$ \mbox{verify $+$ collect}
      $\rightarrow$ tail-only exact refinement};
  \end{tikzpicture}%
  }
  \caption{Why production radix selection rereads the score row.
    The first histogram consumes all $L$ keys before the
    rank boundary exists; classification therefore rereads the materialized row.
    Later round fusion can overlap work only after that boundary exists.
    The lower ledger separates this structural dependency from atomics and synchronization that
    determine the throughput of each pass, and states the execution ordering investigated in
    \cref{sec:insight}.}
  \label{fig:bg:radix-path}
\end{figure}

\FloatBarrier

\paragraph{Execution pressure.}
Each full-row histogram maps $L$ keys onto shared-memory bin counters, followed by a CTA-wide
prefix scan and barrier; classification then funnels qualifying elements through compaction
counters. The production path therefore combines repeated score processing with shared-memory
atomic contention and synchronization~\cite{gvr2026}. Digit width, bin layout, and counter
aggregation can improve the throughput of these stages, but they do not move the first boundary
earlier and cannot remove the structural reread. This distinction makes full-row traversals the
stable cross-implementation axis, while atomics and barriers explain why the cost of each
traversal remains architecture-sensitive.

The production bottleneck is therefore one of boundary timing: approaching the one-scan
floor requires an actionable upper-tail boundary before the mandatory row inspection, without
relaxing the deployment contract of \cref{sec:bg:interface}.  \Cref{sec:insight} tests whether
production score rows supply the required coarse localization, and \cref{sec:method} turns the
resulting evidence into an exact selector.
\FloatBarrier

\section{Motivation}
\label{sec:insight}

\subsection{Empirical Observations}
\label{sec:insight:observations}

Exact \topk selection must resolve the precise $K$th boundary, whereas an early filtering
boundary need only retain a compact superset of the answer.  This distinction raises two
workload questions: how much boundary error can be tolerated before the retained tail ceases to
be compact, and whether a regular partial view can locate that tail across ragged score rows.

We answer these questions using indexer scores produced by Hy4-Preview~\cite{hy4preview2026} and
captured immediately after the indexer MQA and before \topk selection. The measurements use
$K=2048$. Each capture materializes a batch of these FP32 score rows immediately before
selection.  For row $r$, we analyze the valid vector $\mathbf x_r=(x_{r,0},\ldots,x_{r,L_r-1})$.
When one row is clear from context, we omit $r$ and write $\mathbf x=(x_0,\ldots,x_{L-1})$.
Three captures span multiple indexer layers and prefill chunks, with maximum valid lengths of
$65{,}536$, $131{,}072$, and $244{,}650$.  In each capture, we retain rows with $L_r\ge4K$, so
every row supports a common upper-tail range, and select $4{,}096$ rows at equal spacing from
the eligible set.  We preserve their causal lengths and rank all valid FP32 values, so every
quantity below is an exact order statistic of the operator input.

\paragraph{Observation 1: large boundary error can leave only a compact upper tail.}
\label{sec:insight:volume}
Ordinary subscripts continue to identify positions.  To distinguish ranks from positions,
let $x_{(j)}^{\downarrow}$ denote the $j$-th largest value, so
$x_{(1)}^{\downarrow}\ge\cdots\ge x_{(L)}^{\downarrow}$.  An exact selector must ultimately
distinguish $x_{(K)}^{\downarrow}$ from $x_{(K+1)}^{\downarrow}$.
A coarse boundary has a different task: it need only retain a manageable upper tail that contains
the final answer.  We measure the accuracy--work relation directly.  Let
$\delta=x_{(K)}^{\downarrow}-x_{(K+1)}^{\downarrow}$
be the adjacent score gap demanded by exact selection.  We move the boundary downward by $\epsilon$
such gaps and count the complete-row values it retains:
\begin{equation}
  \theta(\epsilon)=x_{(K)}^{\downarrow}-\epsilon\delta,
  \qquad
  C(\epsilon)=\#\{i\in[L]:x_i\ge\theta(\epsilon)\}.
  \label{eq:insight:score-cost}
\end{equation}
Here $\epsilon$ measures boundary error in units of the exact adjacent gap.  $C(\epsilon)$ is the
exact number of complete-row values at or above $\theta(\epsilon)$, namely
the upper-tail size that remains after applying this coarse boundary.  Together they expose the
accuracy--work frontier:
$\epsilon$ controls boundary imprecision, and $C(\epsilon)$ records its cost in values that remain
to be resolved.

Adjacent-gap normalization requires $\delta>0$, so rows tied at the $K$th boundary are omitted
from this sweep.  \Cref{fig:motivation:structure}a varies $\epsilon$ over four orders of
magnitude.  A boundary displaced by $100\times$ the exact gap still retains a median of $1.03K$
values, and $95\%$ of rows retain no more than $1.15K$.  At $500\times$, the median is $1.18K$
and the $95$th percentile is $2.02K$.  The curve eventually bends upward, making the cost of
excessive error visible, but its long flat region is the important separation: score precision
can be relaxed by hundreds of adjacent-boundary gaps before retained work grows by the same
scale.  All three captures show the same flat-then-rising profile, so this tolerance persists
across the measured sequence lengths rather than appearing only at one favorable length.

The reason is local rank geometry.  One exact gap spans one adjacent pair, whereas lowering a
boundary through many comparable gaps advances through many ranks.  Under a smooth local score
density, a displacement of $\epsilon\delta$ adds approximately $\epsilon$ values, not a factor
of $\epsilon$ in retained work.  \Cref{sec:insight:theory} states this approximation precisely.
The experiment therefore identifies a broad target: a useful early signal need not approximate
the exact cutoff; it need only land within the flat part of the score-error/work curve.

\paragraph{Observation 2: regular partial views recover global upper-tail quantiles.}
The first observation supplies tolerance but not a cheap source of global information.  We next
ask whether regularly spaced positions preserve the complete-row upper-tail quantiles.  For a
period $s$, positions with the same residue modulo $s$ form one phase, and each phase reads
approximately a $1/s$ fraction of the row.  For a requested complete-row rank $\widetilde K$,
let $q=\lceil\widetilde K/s\rceil$ be the corresponding sample rank.  In phase $a$, we take its
$q$-th largest value and write $J_a(q)$ for that value's exact rank in the complete row.  A
calibrated view therefore satisfies $J_a(q)\approx\widetilde K$.

We sweep five fractions from $1/16$ through $1/256$, five requested ranks from $K$ through $2K$,
and every phase of the selected rows.  In \cref{fig:motivation:structure}b, their median
complete-row ranks follow the requested-rank diagonal throughout the sweep.  Even the smallest
$1/256$ view remains within $5\%$ of every requested rank, while views through $1/64$ remain
within $1.3\%$.

The median curves establish calibration but do not hide its uncertainty.  The two distribution
glyphs in the same panel show the central shape and median, while their capped spines expose the
fifth-to-$95$th-percentile range at requested rank $1.5K$: $1/64$ views span $1.23K$--$1.77K$,
while $1/128$ views span $1.09K$--$1.92K$.  Their medians remain within $0.03K$ of the request,
and the full uncensored distributions place $95\%$ of absolute errors within $0.32K$ and
$0.50K$, respectively.  The partial views are therefore globally centered, while their remaining
dispersion stays on the same $O(K)$ scale as the compact tail identified by Observation~1.

This representativeness has a position-agnostic basis.  If global ranks are exchangeable over
positions, each phase is a uniform finite-row sample and its order statistic is centered near
the requested global rank.  The exchangeable curve in panel~(d) evaluates the corresponding
finite-population RMS error at every measured causal length; this reference already predicts
useful calibration for all five view sizes.

The captured rows supply additional empirical margin.  Across the three captures, the first and
last position deciles contain $52$--$74\%$ of exact top-$K$ indices, compared with $20\%$ under
uniform positions (panel~(c)).  These endpoint-structured captured rows exhibit $14$--$28\%$
lower RMS rank error than the exchangeable reference at every measured view size (panel~(d)).
Thus regular partial views work in the general position-agnostic case, and the structured real
inputs are easier still.  The endpoint geometry itself is consistent with prior attention-sink
and recency analyses~\cite{xiao2024streamingllm,gu2025sink,barbero2025firsttoken}.

Together, the measurements establish both the tolerance available to a coarse boundary and a
low-cost source of global rank information.  The next subsection gives these two empirical
relations a theoretical basis.

\begin{figure}[!htbp]
  \centering
  \includegraphics[width=\linewidth]{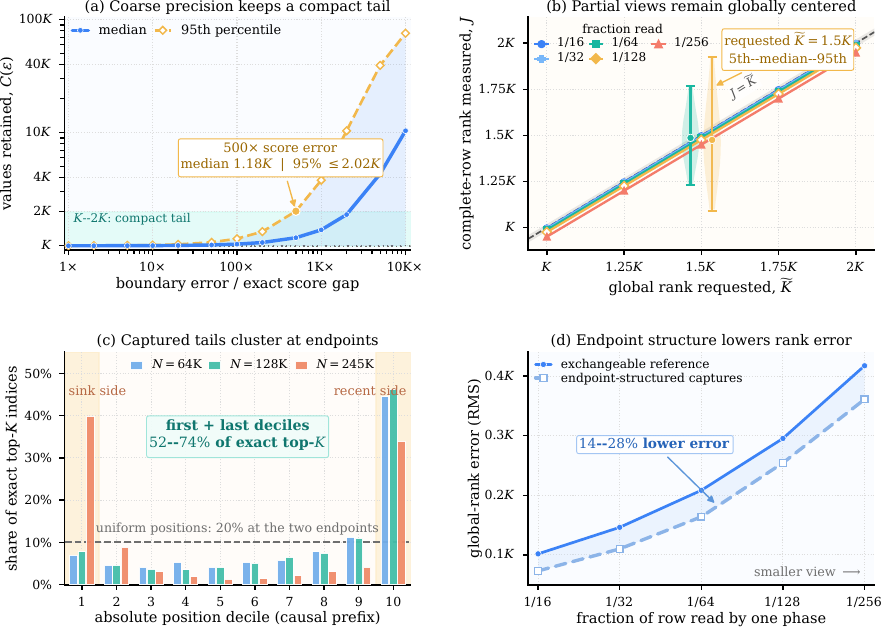}
  \caption{Method-independent evidence that a regular partial view provides a useful coarse
    upper-tail locator.
    (a) Complete-row work after lowering the exact boundary by $\epsilon$ adjacent score gaps; the
    shading extends from the exact $K$ floor to the $95$th percentile.
    (b) Median measured rank follows the requested global rank for five read fractions; the two
    distribution glyphs combine the observed central shape, median, and fifth-to-$95$th-percentile
    range for $1/64$ and $1/128$ views at $1.5K$.
    (c) Exact top-$K$ indices in all three captures concentrate in the first and last absolute
    position deciles, the broad endpoint geometry associated with sink and recency regions.
    (d) At requested rank $1.5K$, endpoint-structured captured rows
    have $14$--$28\%$ lower RMS
    rank error than the exact position-exchangeable reference across all five read fractions.}
  \label{fig:motivation:structure}
\end{figure}
\FloatBarrier

\subsection{Theoretical Basis}
\label{sec:insight:theory}

The observations expose two distinct resolutions.  In value space, the issue is how a score
displacement around the exact boundary translates into additional retained ranks.  In position
space, the issue is how an order statistic from a regular partial view maps back to the complete
row.  Local and finite-population order statistics make both relations explicit, while
contiguous-interval arithmetic explains the additional margin supplied by the observed endpoint
structure.

\paragraph{Local order statistics produce the flat cost region.}
Let $F$ be a smooth local approximation to a row's score distribution, $f$ its density, and
$\xi(u)=F^{-1}(1-u)$ its upper quantile function.  Using
$x_{(j)}^{\downarrow}\approx\xi(j/L)$, a first-order expansion around rank $K$
gives~\cite{davidnagaraja2003}
\begin{equation}
  \delta=x_{(K)}^{\downarrow}-x_{(K+1)}^{\downarrow}
  \approx\frac{1}{Lf(x_{(K)}^{\downarrow})},
  \qquad
  x_{(K)}^{\downarrow}-x_{(C)}^{\downarrow}
  \approx\frac{C-K}{Lf(x_{(K)}^{\downarrow})}.
  \label{eq:insight:local-spacing}
\end{equation}
For $C=C(\epsilon)$, the threshold definition in \cref{eq:insight:score-cost} implies
$x_{(C)}^{\downarrow}\ge\theta(\epsilon)>x_{(C+1)}^{\downarrow}$.  Thus $\epsilon\delta$ lies
between the score displacements to ranks $C$ and $C+1$.  Applying the two
local-spacing approximations yields, up to this one-rank discretization,
\begin{equation}
  C(\epsilon)-K\approx\epsilon.
  \label{eq:insight:score-to-rank}
\end{equation}
The unknown local density and row length cancel.  At $\epsilon=500$, for example, the relation
predicts a tail near $K+500\approx1.24K$, close to the measured $1.18K$ median.  The empirical ribbon then
captures the remaining variation in local score geometry.  This rank-scale interpretation explains
the flat opening of \cref{fig:motivation:structure}a: hundreds of exact score gaps add hundreds of
values to an existing $K$-sized tail, rather than multiplying its size by hundreds.

\paragraph{Exchangeability supplies the position-agnostic baseline.}
For a valid row of length $L$, phase $a$ contains
$n_a\in\{\lfloor L/s\rfloor,\lceil L/s\rceil\}$ positions.  Under positional exchangeability it is
a uniform $n_a$-subset of the global rank order.  The complete-row rank $J_a(q)$ of its $q$-th
largest value is a finite-population order statistic~\cite{oneill2022orderstat,davidnagaraja2003}
with
\begin{equation}
  \mathbb{E}[J_a(q)]=\frac{q(L+1)}{n_a+1}\approx qs\approx\widetilde K,
  \qquad
  \operatorname{Var}[J_a(q)]
  =\frac{q(n_a-q+1)(L+1)(L-n_a)}{(n_a+1)^2(n_a+2)}.
  \label{eq:insight:exchangeable}
\end{equation}
Because $n_a\approx L/s$ and $q=\lceil\widetilde K/s\rceil$, the expectation reduces to the
requested global rank up to the one-sample rounding margin, giving the diagonal in
\cref{fig:motivation:structure}b.  The RMS error around that request is
$\sqrt{\operatorname{Var}[J_a(q)]+(\mathbb E[J_a(q)]-\widetilde K)^2}$, so the exact reference
retains both the finite-row variance and the small rounding bias.  Ignoring these edge effects when
$L\gg\widetilde K$, the variance reduces to approximately
$\widetilde K(s-1)(1-\widetilde K/L)$; shrinking the view therefore increases uncertainty only at a
square-root rate.  Evaluating the exact expression for every measured causal length and both
possible phase sizes produces the exchangeable RMS reference in panel~(d), rather than a curve
fitted to the measurements.

\paragraph{Contiguous boundary regions distribute themselves across phases.}
For any contiguous boundary region $\mathcal B\subseteq[L]$ of $m$ positions,
\begin{equation}
  \#\{i\in\mathcal B:i\bmod s=a\}\in
  \left\{\left\lfloor\frac{m}{s}\right\rfloor,
          \left\lceil\frac{m}{s}\right\rceil\right\}.
  \label{eq:insight:interval}
\end{equation}
Writing $m=\ell s+\nu$ with $0\le\nu<s$ makes the balance explicit: every residue occurs
$\ell$ times in the region and exactly $\nu$ residues occur once more.
Hence counts across phases differ by at most one.  When elevated upper-tail probability extends
across broad endpoint regions, as measured in \cref{fig:motivation:structure}c, every phase receives
nearly the same number of opportunities to observe it.  This shared coverage can suppress
phase-to-phase rank dispersion; a tail locked to one residue instead creates the classical
systematic-sampling alias~\cite{bellhouse1988systematic}.  The interval identity therefore supplies
the geometric mechanism behind the measured reduction, while
panel~(d) measures its magnitude.

\subsection{Design Implications}
\label{sec:insight:consequence}

The evidence supports a coarse-to-exact decomposition.  Observation~1 shows that locating a
compact upper tail requires far less score resolution than identifying the final $K$th boundary.
Observation~2 shows that a regular partial view supplies a globally centered estimate at
precisely this coarser rank scale, with the observed endpoint structure adding margin rather
than defining the applicability of the estimate.  An early estimator should therefore be judged
by the work it leaves for exact resolution, not by whether it reproduces the final cutoff.

This division assigns the two stages different responsibilities.  The partial view provides an
early boundary with enough global context to organize the upper tail.  The complete-row stage
remains authoritative: it sees every valid value, verifies that the retained tail is sufficient,
and leaves the final ordering to an exact operation over that compact set.  Sampling error can
change the amount of retained work, but the complete-row check prevents it from changing the
final output.

The decomposition supplies the early boundary required by \cref{sec:bg:cuda-path}.  The
mandatory pass can classify and collect values on their first encounter, while subsequent exact
work remains confined to the upper tail.  \Cref{sec:method} develops this principle into a
sample-guided exact selector.
\FloatBarrier

\section{Method}
\label{sec:method}

\subsection{Algorithm Overview}
\label{sec:method:overview}

Consider one valid score row $\mathbf x=(x_0,\ldots,x_{L-1})\in\mathbb R^L$, with $L>K$.  Rows
are selected independently, so the three phases below operate on this row. HPC-Ops separates
boundary proposal from exact selection.  Phase~1 reads a regular partial view and proposes a
row-local coarse boundary.  Phase~2 performs the authoritative complete-row traversal, using
that boundary to persist candidates, count them, and construct the first exact histogram on
their first encounter.  Phase~3 consumes that histogram, commits resolved digit prefixes, and
carries only the unique boundary group through the remaining FP32 digits.  An underfilled
proposal may first expand through a bounded secondary boundary before exact coarse recovery;
rows without a usable view enter exact coarse recovery directly.

\begin{figure}[!htbp]
  \centering
  \begin{tikzpicture}[
    font=\footnotesize,
    phase/.style={
      rounded corners=2.5pt,
      minimum height=1.48cm,
      align=center,
      inner sep=4pt,
      line width=0.75pt
    },
    pOne/.style={phase,draw=HYDarkBlue!68,fill=HYLightBlue!25},
    pTwo/.style={phase,draw=RoyalBlue!72,fill=RoyalBlue!7},
    pThree/.style={phase,draw=ForestGreen!68!black,fill=ForestGreen!7},
    recover/.style={
      rounded corners=2pt,
      minimum height=0.78cm,
      align=center,
      inner sep=3.5pt,
      font=\scriptsize,
      line width=0.68pt,
      draw=orange!75!black,
      fill=orange!7
    },
    fallback/.style={recover,draw=BrickRed!70,fill=BrickRed!5},
    certified/.style={recover,draw=ForestGreen!65!black,fill=ForestGreen!7},
    mapbox/.style={
      rounded corners=3pt,
      draw=HYDarkBlue!35,
      fill=black!1,
      minimum height=2.42cm,
      inner sep=4pt,
      line width=0.62pt
    },
    miniOne/.style={rounded corners=1.5pt,draw=HYDarkBlue!62,fill=HYLightBlue!28,
      minimum height=0.54cm,align=center,font=\scriptsize,inner sep=2pt},
    miniTwo/.style={rounded corners=1.5pt,draw=RoyalBlue!65,fill=RoyalBlue!8,
      minimum height=0.54cm,align=center,font=\scriptsize,inner sep=2pt},
    miniThree/.style={rounded corners=1.5pt,draw=ForestGreen!65!black,fill=ForestGreen!7,
      minimum height=0.54cm,align=center,font=\scriptsize,inner sep=2pt},
    flow/.style={-{Latex[length=2.0mm]},line width=0.78pt,draw=HYDarkBlue!80},
    exact/.style={-{Latex[length=1.9mm]},line width=0.75pt,draw=ForestGreen!65!black},
    branch/.style={-{Latex[length=1.8mm]},line width=0.72pt,draw=orange!78!black},
    fail/.style={-{Latex[length=1.8mm]},line width=0.72pt,draw=BrickRed!75},
    thinflow/.style={-{Latex[length=1.5mm]},line width=0.62pt,draw=HYDarkBlue!65}
  ]
    \node[pOne,text width=2.65cm] (p1) at (1.50,6.12) {
      \textcolor{HYDarkBlue}{\textbf{Phase 1}}\\[-1pt]
      \textbf{Coarse boundary}\\[-1pt]
      \textbf{localization}\\[2pt]
      row-local view $\mathcal V_{r,s}$ $\rightarrow$ histogram\\
      provisional $\tau_0$ and nested $\tau_1$};
    \node[pTwo,text width=4.15cm] (p2) at (6.68,6.12) {
      \textcolor{RoyalBlue!82!black}{\textbf{Phase 2}}\\[-1pt]
      \textbf{Fused validation + formation}\\[2pt]
      one complete-row scan $\rightarrow$\\[-1pt]
      $\mathcal A(\tau_0),\ C(\tau_0)$\\[-1pt]
      $\widehat h_0(\cdot;\tau_0)$};
    \node[pThree,text width=3.45cm] (p3) at (13.20,6.12) {
      \textcolor{ForestGreen!62!black}{\textbf{Phase 3}}\\[-1pt]
      \textbf{Exact radix refinement}\\[2pt]
      $\mathcal Q_j\rightarrow b_j$ by digit prefix\\
      commit $\mathcal R_j$; retain $\mathcal Q_{j+1}$};

    \draw[flow] (p1.east) -- node[midway,above=2.5pt,font=\scriptsize,
      text=HYDarkBlue!80,inner sep=0pt]
      {propose $\tau_0$} (p2.west);
    \draw[exact] (p2.east) --
      node[midway,above=2.5pt,font=\scriptsize,text=ForestGreen!58!black,
        inner sep=0pt]
        {$(\mathcal E_0,\mathcal Q_0,k_0,h_0)$}
      node[midway,below=3.0pt,font=\scriptsize,text=ForestGreen!58!black,
        inner sep=0pt]
        {$C(\tau_0)\ge K$}
      (p3.west);

    \node[recover,text width=4.24cm] (band) at (6.68,4.24) {
      \textbf{Bounded band extension}\\[-1pt]
      append $\mathcal A(\tau_1)\setminus\mathcal A(\tau_0)$};
    \node[fallback,text width=3.34cm] (full) at (6.68,3.03) {
      \textbf{Exact coarse recovery}\\[-1pt]
      rebuild boundary from all $L$ values};
    \node[certified,text width=3.28cm] (cert) at (12.08,3.62) {
      \textbf{certified interface}\\[-1pt]
      $(\mathcal E_0,\mathcal Q_0,k_0,h_0)$};

    \draw[branch] (p2.south) -- node[right,font=\scriptsize,text=orange!72!black]
      {$C(\tau_0)<K$} (band.north);
    \draw[exact] (band.east) -- node[pos=0.58,above=3.3pt,inner sep=0pt,
      font=\scriptsize,text=ForestGreen!58!black] {enough} (cert.west);
    \draw[fail] (band.south) -- node[right,font=\scriptsize,text=BrickRed!72]
      {still short} (full.north);
    \draw[dashed,line width=0.72pt,draw=BrickRed!75] (p1.south) -- (1.50,3.03);
    \draw[-{Latex[length=1.8mm]},dashed,line width=0.72pt,draw=BrickRed!75]
      (1.50,3.03) -- node[midway,above=2.5pt,font=\scriptsize,text=BrickRed!72,
      inner sep=0pt] {no usable view} (full.west);
    \draw[exact] (full.east) -- (cert.south west);
    \draw[exact] (cert.north east) to[out=46,in=-104] (p3.south);

    \draw[HYDarkBlue!24,line width=0.6pt] (0,2.30) -- (15.1,2.30);
    \node[font=\fontsize{9.5}{10.5}\selectfont\bfseries,text=HYDarkBlue,anchor=west]
      at (0,2.00)
      {Sampled execution mappings};
    \node[font=\scriptsize,text=black!60,anchor=east] at (15.1,2.02)
      {same phases and invariant; different exposed parallelism};

    \node[mapbox,text width=5.62cm] (onebox) at (3.02,0.54) {};
    \node[font=\footnotesize\bfseries,text=HYDarkBlue,anchor=west]
      at (0.35,1.52) {Persistent one CTA per row};
    \node[miniOne,text width=1.28cm] (o1) at (1.15,0.68)
      {Phase 1\\locate $\tau_0$};
    \node[miniTwo,text width=1.42cm] (o2) at (3.02,0.68) {Phase 2\\scan row};
    \node[miniThree,text width=1.28cm] (o3) at (4.91,0.68)
      {Phase 3\\$\mathcal Q_j\!\rightarrow\!\mathcal Q_{j+1}$};
    \draw[thinflow] (o1) -- (o2);
    \draw[thinflow] (o2) -- (o3);
    \node[font=\scriptsize,text=black!63,align=center] at (3.02,-0.32)
      {phase state remains local; independent rows\\are taken in completion order};

    \node[mapbox,text width=8.44cm] (splitbox) at (10.76,0.54) {};
    \node[font=\footnotesize\bfseries,text=HYDarkBlue,anchor=west]
      at (6.72,1.52) {KV-split for a few long rows};
    \node[font=\scriptsize\bfseries,text=RoyalBlue!80!black,align=center]
      at (9.38,1.22) {split Phase 2};

    \node[miniOne,text width=1.34cm,minimum height=0.72cm]
      (splitP1) at (7.58,0.56) {Phase 1\\locate $\tau_0$};

    \node[miniTwo,text width=1.42cm,minimum height=0.30cm] (s20) at (9.38,0.90)
      {scan $\Omega_0$};
    \node[miniTwo,text width=1.42cm,minimum height=0.30cm] (s21) at (9.38,0.56)
      {scan $\Omega_1$};
    \node[miniTwo,text width=1.42cm,minimum height=0.30cm] (s22) at (9.38,0.22)
      {scan $\Omega_{P-1}$};

    \node[rounded corners=1.5pt,draw=orange!70!black,fill=orange!6,
      text width=1.62cm,minimum height=1.22cm,align=center,font=\scriptsize]
      (reduce) at (11.70,0.56) {merge\\$\biguplus_p\mathcal A_p$\\$\sum_p C_p(\tau)$\\
      $\sum_p\widehat h_{0,p}$};
    \node[miniThree,text width=1.18cm,minimum height=1.22cm] (finish)
      at (14.02,0.56) {Phase 3\\radix\\refine};

    \draw[thinflow] (splitP1.east) to[out=0,in=180] (s20.west);
    \draw[thinflow] (splitP1.east) -- (s21.west);
    \draw[thinflow] (splitP1.east) to[out=0,in=180] (s22.west);
    \draw[thinflow] (s20.east) to[out=0,in=157] (reduce.west);
    \draw[thinflow] (s21.east) -- (reduce.west);
    \draw[thinflow] (s22.east) to[out=0,in=203] (reduce.west);
    \draw[thinflow] (reduce) -- (finish);
    \node[font=\scriptsize,text=black!63,align=center] at (10.76,-0.45)
      {row-local boundary \;\textbullet\; disjoint scans \;\textbullet\;
       additive merge \;\textbullet\; one finisher};
  \end{tikzpicture}
  \caption{The sample-guided exact-selection pipeline and its two
    sampled execution mappings.
    The upper path separates boundary proposal, complete-row certification, and exact FP32
    refinement; an underfilled proposal expands through a precomputed band and otherwise enters an
    exact coarse recovery.  Phase~2 converts either a certified proposal histogram or
    an exact coarse recovery into the common
    $(\mathcal E_0,\mathcal Q_0,k_0,h_0)$ refinement interface.  In Phase~3, each digit commits
    strictly better groups and advances only the unique boundary group.  The lower path maps the
    same phases either to one CTA per row or to a row-local boundary, partitioned validation,
    additive merging, and one-finisher refinement for KV-split execution.}
  \label{fig:method:overview}
\end{figure}
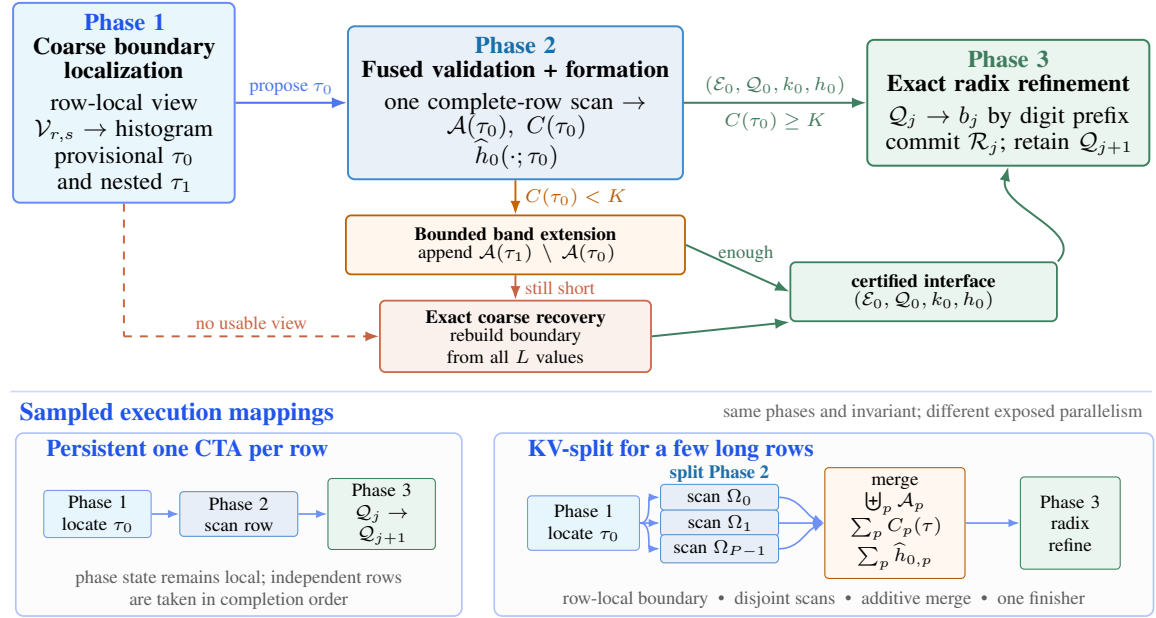

\FloatBarrier

This organization lets the mandatory complete-row pass validate the proposed boundary and form
its admitted tail together.  The partial view determines how much work remains, whereas the
complete row and the FP32 refinement determine the returned indices.

Every route entering Phase~3 produces the same initialization
$(\mathcal E_0,\mathcal Q_0,k_0,h_0)$: a committed prefix, one unresolved frontier, the remaining
quota, and the frontier's leading-digit histogram.  Its certified admitted state is
$\mathcal S_0:=\mathcal E_0\uplus\mathcal Q_0\subseteq[L]$, satisfying
\begin{equation}
  |\mathcal S_0|\ge K,
  \qquad
  i\in\mathcal S_0,\ j\notin\mathcal S_0
  \quad\Longrightarrow\quad x_i\ge x_j.
  \label{eq:method:certified-invariant}
\end{equation}
The boundary construction guarantees the ordering relation, and the complete-row count establishes
its sufficiency.  Exact refinement does not order this state wholesale.  At each FP32 digit, it
commits every group strictly ahead of the rank boundary and carries only the unique boundary group
to the next digit.  A conservative proposal therefore changes only the initial refinement volume;
an underfilled proposal is detected before any output is committed.

\Cref{fig:method:overview} summarizes the logical pipeline and its two sampled execution
mappings. The three phases preserve the same certified-state invariant whether one CTA owns a
row or several CTAs partition its complete-row traversal.

\subsection{Phase 1: Coarse Boundary Localization}
\label{sec:method:estimate}

Phase~1 locates its boundary in an inexpensive coarse ordering.
Let $\operatorname{fp16}(x)$ denote round-to-nearest FP16 projection, and let
$\kappa_{16}^{\downarrow}(x)$ be the descending sortable key of that projection, so a larger
projected value receives a smaller integer key.
Grouping adjacent keys by their 11 most significant bits gives
\begin{equation}
  c(x)=\left\lfloor
    \frac{\kappa_{16}^{\downarrow}(x)}{2^5}
  \right\rfloor
  \in\{0,\ldots,B_c-1\},
  \qquad B_c=2^{11}.
  \label{eq:method:coarse-map}
\end{equation}
This coarse map is monotone and many-to-one:
\begin{equation}
  x_i>x_j \quad\Longrightarrow\quad c(x_i)\le c(x_j).
  \label{eq:method:coarse-order}
\end{equation}
A larger score may share a bin with a smaller one but cannot enter a worse bin.  Phase~3 later
resolves the order within each coarse bin from the original FP32 values.

Across rows, we distribute the regular views with a balanced deterministic phase schedule.  Let
$\pi_s$ be a fixed permutation of the $s$ residue classes and assign row $r$ the phase
$a_r=\pi_s(r\bmod s)$.  For its valid length $L_r$, the locator reads
\begin{equation}
  \mathcal{V}_{r,s}=\{i\in[L_r]:i\equiv a_r\pmod s\}.
  \label{eq:method:view}
\end{equation}
Every block of $s$ rows consequently uses each phase once.  This removes a dataset-wide preference
for any one residue class while preserving regular, deterministic access within a row.  It also
allows cooperating execution units to reconstruct the same view from row identity alone.  The
view contains approximately $L_r/s$ values, enough to expose an upper-tail order statistic while
remaining a fraction of one complete traversal.

The sampled values form a coarse-key histogram and inclusive prefix count
\begin{equation}
  h_{r,s}(b)=\sum_{i\in\mathcal{V}_{r,s}}\mathbf{1}[c(x_{r,i})=b],
  \qquad
  H_{r,s}(b)=\sum_{v=0}^{b}h_{r,s}(v).
  \label{eq:method:sample-hist}
\end{equation}

Localization targets the same complete-row rank used in Observation~2.  Let
$\widetilde K=\rho K$, where $\rho>1$ supplies a compact retention margin, and request sample rank
$q=\lceil\widetilde K/s\rceil$.
The provisional boundary for row $r$ is the first coarse bin whose prefix contains that sample rank,
\begin{equation}
  \tau_{0,r}=\min\{b:H_{r,s}(b)\ge q\}.
  \label{eq:method:primary-threshold}
\end{equation}
Each row therefore derives its own bin threshold $\tau_{0,r}$ while sharing the same global-rank
target.  The relationship in \cref{sec:insight:observations,sec:insight:theory} explains why the
resulting boundary rank is centered near $\widetilde K$, rather than growing with $L_r$.  We choose
$\widetilde K=1.5K$ in the evaluated configuration, within the
compact $O(K)$ tail established by
Observation~1 and the rank range evaluated by Observation~2.  The margin favors a compact superset
over a boundary that stops at exactly rank $K$.

A deeper target $\widetilde K_1>\widetilde K$ may be recorded during the same prefix scan, with
$q_1=\lceil\widetilde K_1/s\rceil$ and $\tau_{1,r}\ge\tau_{0,r}$. The primary boundary defines
the common path; the secondary boundary remains dormant unless the complete-row count exposes an
underfilled proposal.  Both boundaries come from the same histogram, so this recovery margin
requires neither a denser view nor another sample traversal.  Phase~1 therefore returns a
row-local primary boundary and, when enabled, a nested recovery boundary. The row subscript is
omitted from $\tau_0$ and $\tau_1$ below when the row is clear; rows whose views cannot support
the requested sample rank enter the exact route directly.

These row-local boundaries determine the downstream candidate volume and whether secondary or
exact recovery is needed.  Phase~2 materializes the corresponding admitted region over all valid
scores and certifies its sufficiency; Phase~3 resolves the exact FP32 order.

\FloatBarrier
\subsection{Phase 2: Fused Validation and Candidate Formation}
\label{sec:method:validate}

For a coarse threshold $\tau$, define its admitted upper tail and size as
\begin{equation}
  \mathcal{A}(\tau)=\{i\in[L]:c(x_i)\le\tau\},
  \qquad
  C(\tau)=|\mathcal{A}(\tau)|.
  \label{eq:method:candidate-set}
\end{equation}
Because $c$ is monotone, the admitted prefix $c(x)\le\tau$ forms a contiguous upper tail in FP32
score space.  Its cutoff is the smallest FP32 value admitted by that prefix.
The inclusive boundary admits an entire coarse bin, so reduced precision may enlarge the retained
tail but cannot cut through that bin.  This $C(\tau)$ is the algorithmic counterpart of
$C(\epsilon)$ in \cref{eq:insight:score-cost}: both denote retained complete-row work, with their
arguments identifying how the threshold is chosen.

\begin{lemma}[Validated Top-$K$ Containment]
\label[lemma]{lem:method:containment}
If $C(\tau)\ge K$, then $\mathcal{A}(\tau)$ contains a Top-$K$ solution for $\mathbf{x}$.
\end{lemma}

\begin{proof}
For any admitted $i$ and excluded $j$, $c(x_i)\le\tau<c(x_j)$.  If $x_j>x_i$, the monotonicity in
\cref{eq:method:coarse-order} would require $c(x_j)\le c(x_i)$, a contradiction.  Hence every
admitted value is no smaller than every excluded value.  Because the complete boundary bin is
admitted, the $K$ largest values inside $\mathcal{A}(\tau)$ form a valid exact Top-$K$ solution.
\end{proof}

For exact refinement, let $\kappa_{32}^{\downarrow}(x)$ be the descending sortable key of
the original FP32 score.  We partition its 32 bits into $D=3$ digits of widths $11+11+10$ and write
$\operatorname{dig}_j(\kappa_{32}^{\downarrow}(x))$ for digit $j$.
Before the complete-row traversal, Phase~2 converts $\tau_0$ once to this FP32 cutoff.  It then
compares the original FP32 scores directly against the cutoff while materializing
$\mathcal A(\tau_0)$, storing its indices, and measuring $C(\tau_0)$.  On the same encounter, each
admitted FP32 value contributes to a
proposal histogram
over the leading exact digit, producing
\begin{equation}
  \widehat h_0(b;\tau_0)
  =\sum_{i\in\mathcal{A}(\tau_0)}
    \mathbf{1}[\operatorname{dig}_0(\kappa_{32}^{\downarrow}(x_i))=b].
  \label{eq:method:first-exact-hist}
\end{equation}
The hat records that this histogram belongs to a proposed admitted tail; the unadorned $h_0$ is
reserved below for the certified frontier passed to Phase~3.  The pass therefore produces a
persisted admitted tail $\mathcal{A}(\tau_0)$, its containment certificate $C(\tau_0)$, and its
leading exact-digit histogram.  On the primary path this tail becomes the initial FP32 refinement
frontier.  Candidate formation and the proposal histogram are fused because they inspect the same
admitted values.  Unlike the conventional path in \cref{fig:bg:radix-path}, no complete traversal
is spent discovering a boundary that can only be used after the row has passed; the provisional
boundary already exists when the scan starts.

If $C(\tau_0)=K$, the certified tail is already the answer; if $C(\tau_0)>K$, it becomes the
initial FP32 refinement frontier.  Estimation error on the loose side increases $C$ and hence
refinement work, while error on the tight side is detected as underfill.

\paragraph{Underfilled views.}
When $C(\tau_0)<K$, the nested boundary from Phase~1 permits a bounded recovery before rebuilding an
exact boundary.  A second scan appends only the band
\begin{equation}
  \Delta\mathcal{A}
  =\mathcal{A}(\tau_1)\setminus\mathcal{A}(\tau_0)
  =\{i:\tau_0<c(x_i)\le\tau_1\}
  \label{eq:method:rescue-band}
\end{equation}
and extends the proposal histogram by
\begin{equation}
  \widehat h_0(b;\tau_1)
  =\widehat h_0(b;\tau_0)
   +\sum_{i\in\Delta\mathcal A}
    \mathbf{1}[\operatorname{dig}_0(\kappa_{32}^{\downarrow}(x_i))=b].
  \label{eq:method:rescue-hist}
\end{equation}
The enlarged set is certified by the same complete-row count.  Once sufficient,
$\mathcal A(\tau_1)$ replaces $\mathcal A(\tau_0)$ as the initial refinement frontier; only
continued underfill invokes exact coarse recovery.  The secondary threshold therefore provides a
bounded expansion of the candidate superset without changing the exact-selection criterion.

If no secondary boundary is enabled, or if the enlarged set remains underfilled, the selector
enters exact coarse recovery.  It builds the coarse histogram over all $L$ values and identifies
$\tau_\star=\min\{b:C(b)\ge K\}$ as the exact coarse boundary.
A subsequent traversal commits every bin strictly ahead of $\tau_\star$, retains the complete
boundary bin, discards the remaining bins, and constructs the leading exact-digit histogram only
over that retained frontier.  The resulting initialization is
\begin{equation}
  \mathcal E_0=\{i:c(x_i)<\tau_\star\},\;
  \mathcal Q_0=\{i:c(x_i)=\tau_\star\},\;
  k_0=K-|\mathcal E_0|,\;
  h_0(b)=\sum_{i\in\mathcal Q_0}
    \mathbf{1}[\operatorname{dig}_0(\kappa_{32}^{\downarrow}(x_i))=b].
  \label{eq:method:exact-init}
\end{equation}

Consequently, every route that requires refinement supplies
$(\mathcal E_0,\mathcal Q_0,k_0,h_0)$ with the same semantics.  A certified sampled route sets
\begin{equation}
  \mathcal E_0=\varnothing,
  \qquad
  \mathcal Q_0=\mathcal A(\tau),
  \qquad
  k_0=K,
  \qquad
  h_0(\cdot)=\widehat h_0(\cdot;\tau),
  \label{eq:method:sampled-init}
\end{equation}
where $\tau\in\{\tau_0,\tau_1\}$ is the boundary whose tail passed validation; exact coarse
recovery supplies \cref{eq:method:exact-init}.  In both cases,
$\mathcal S_0=\mathcal E_0\uplus\mathcal Q_0$ satisfies
\cref{eq:method:certified-invariant}.  This route-independent tuple is the interface between
validation and refinement.

\subsection{Phase 3: Exact Radix Refinement}
\label{sec:method:refine}

Phase~3 consumes $(\mathcal E_0,\mathcal Q_0,k_0,h_0)$.  At the beginning of round
$j\in\{0,\ldots,D-1\}$, $\mathcal E_j$ contains the committed indices,
$\mathcal Q_j$ is the unique unresolved frontier, and $k_j=K-|\mathcal E_j|$ is the remaining
quota.  Its histogram and inclusive prefix are
\begin{equation}
  h_j(b)=\sum_{i\in\mathcal{Q}_j}
    \mathbf{1}[\operatorname{dig}_j(\kappa_{32}^{\downarrow}(x_i))=b],
  \qquad
  H_j(b)=\sum_{v=0}^{b}h_j(v).
  \label{eq:method:refine-prefix}
\end{equation}
Because smaller ordered keys denote larger FP32 values, the boundary digit is the first bin whose
prefix reaches the remaining quota:
\begin{equation}
  b_j=\min\{b:H_j(b)\ge k_j\}.
  \label{eq:method:digit-boundary}
\end{equation}
Equivalently, $H_j(b_j-1)<k_j\le H_j(b_j)$, with $H_j(-1)=0$.  Thus every digit group before
$b_j$ is part of the result, while the rank boundary lies inside $b_j$.  The newly committed group is
\begin{equation}
  \mathcal{R}_j
  =\{i\in\mathcal{Q}_j:
    \operatorname{dig}_j(\kappa_{32}^{\downarrow}(x_i))<b_j\}.
  \label{eq:method:committed-group}
\end{equation}
Every index in $\mathcal{R}_j$ is irrevocably selected, every group after $b_j$ is discarded, and
only the boundary group remains unresolved:
\begin{align}
  \mathcal{E}_{j+1}
  &=\mathcal{E}_j\uplus\mathcal{R}_j,
  &
  \mathcal{Q}_{j+1}
  &=\{i\in\mathcal{Q}_j:
    \operatorname{dig}_j(\kappa_{32}^{\downarrow}(x_i))=b_j\},
  &
  k_{j+1}
  &=k_j-|\mathcal{R}_j|.
  \label{eq:method:refine-recurrence}
\end{align}
While classifying $\mathcal{Q}_j$, the algorithm retains $\mathcal{Q}_{j+1}$ and builds its
next-digit histogram.  The next round therefore inspects only the unique boundary group rather than
the complete admitted tail or score row, following a prefix--commit--retain pattern at every digit.
\Cref{sec:impl:refinement} describes the physical buffering of these frontiers.

After all $D$ digits, candidates in $\mathcal Q_D$ have identical FP32 ordered keys. Phase~3
fills the remaining $k_D$ slots from that group. This returns exactly $K$ indices while leaving
the choice among equal-valued boundary positions unspecified under the operator's tie semantics.

\begin{theorem}[End-to-End Exactness]
\label{thm:method:exactness}
For every admissible row with $L>K$, the algorithm returns $K$ distinct valid indices satisfying
the exact Top-$K$ contract of \cref{sec:bg:interface}.
\end{theorem}

\begin{proof}[Proof sketch]
Phase~2 either certifies a sampled admitted tail by \cref{lem:method:containment} or constructs a
sufficient committed-prefix--plus--frontier state with the exact coarse route.
Equations~\eqref{eq:method:digit-boundary}--\eqref{eq:method:refine-recurrence} then commit every
strictly better group while preserving enough indices in the unique boundary frontier to fill the remaining
quota.  The $D=3$ digits exhaust the FP32 ordered key; selecting the remaining indices from the final
equal-key group produces $K$ distinct valid indices, completing the exact Top-$K$ output.
\end{proof}

The proof isolates the role of reduced precision: the FP16 key bounds candidate membership, but
it never resolves the final order.  Exactness is supplied by complete-row certification, exact
FP32 refinement, and the sample-independent recovery path.

\subsection{Execution Mappings}
\label{sec:method:mapping}

For the sampled path, the three logical phases admit two GPU mappings, chosen according to the
row-level parallelism exposed by the call.  Both appear in the lower half of
\cref{fig:method:overview} and preserve the same certification and refinement invariants.

\paragraph{Persistent one CTA per row.}
When the call contains enough rows, one CTA executes Phases~1--3 for a row and retains
phase-local histograms, candidate indices, and the selected prefix across the pipeline.  Rows
remain mutually independent, so a persistent pool can take them in completion order and absorb
variation in valid length and candidate volume.  This mapping minimizes coordination: the only
complete-row work is the Phase~2 scan on a certified sample path, while every exact radix round
stays local to the CTA's frontier state.

\paragraph{KV-split mapping.}
For a small number of long rows, row-level parallelism leaves most of the GPU idle during the
mandatory scan.  We partition $[L]$ into $P$ disjoint contiguous segments
$\Omega_0,\ldots,\Omega_{P-1}$.
The deterministic view lets every participant reconstruct the same row-local boundary, after which
the CTAs partition Phase~2 across disjoint segments, producing
\begin{equation}
  \mathcal{A}_p(\tau)=\mathcal{A}(\tau)\cap\Omega_p,
  \quad C_p(\tau)=|\mathcal{A}_p(\tau)|,
  \quad \widehat h_{0,p}(b;\tau)
  =\sum_{i\in\mathcal{A}_p(\tau)}
    \mathbf{1}[\operatorname{dig}_0(\kappa_{32}^{\downarrow}(x_i))=b].
  \label{eq:method:split-local}
\end{equation}
Because the segments do not overlap, their outputs compose exactly:
\begin{equation}
  \mathcal{A}(\tau)=\biguplus_{p=0}^{P-1}\mathcal{A}_p(\tau),
  \qquad C(\tau)=\sum_{p=0}^{P-1}C_p(\tau),
  \qquad \widehat h_0(b;\tau)=\sum_{p=0}^{P-1}\widehat h_{0,p}(b;\tau).
  \label{eq:method:split-reduce}
\end{equation}
After all segments publish their local state, one finisher reduces
\cref{eq:method:split-reduce}, compacts the disjoint candidate ranges, certifies the aggregate, and
initializes the same $(\mathcal E_0,\mathcal Q_0,k_0,h_0)$ interface used by the one-CTA path.  An
underfilled aggregate enters the same recovery logic.  KV-split therefore changes the parallel
span of the complete-row validation, not the admitted tail or the selection semantics.
The next subsection quantifies the row-level work and reduced validation span introduced by
this mapping.

\subsection{Complexity Analysis}
\label{sec:method:complexity}

On the certified sampled path, Phase~1 reads approximately $L/s$ scores, Phase~2 reads all $L$
scores once, and Phase~3 reads only the successively shrinking boundary frontiers.  Per row, the
work is therefore
\begin{equation}
  \Theta\!\left(\frac{L}{s}+L+
    \sum_{j=0}^{D-1}|\mathcal Q_j|+K\right),
  \qquad D=3.
  \label{eq:method:fast-cost}
\end{equation}
The first two terms describe fixed-coverage passes over the input row: a regular $1/s$ view
costing $L/s$ and one complete certification-and-candidate pass costing $L$.  The frontier sum
accounts for exact work over successively retained sets and therefore depends on the realized
boundary, while the final $K$ term writes the output.  Exactness still requires the complete
$L$-value pass.

\begin{table}[!htbp]
  \centering
  \caption{Score-row traffic before frontier-only FP32 refinement.  A band attempt and the exact
    coarse recovery each add only after the preceding candidate set underfills; all routes terminate
    with the same certified-state invariant.}
  \label{tab:method:cost}
  \small
  \setlength{\tabcolsep}{6pt}
  \renewcommand{\arraystretch}{1.18}
  \begin{tabular}{@{}lcl@{}}
    \toprule
    \textbf{Route} & \textbf{Score reads} & \textbf{Initial FP32 frontier} \\
    \midrule
    Certified primary boundary & $(1+1/s)L$ & admitted tail $\mathcal{A}(\tau_0)$ \\
    Certified after band extension & $(2+1/s)L$ & expanded tail $\mathcal{A}(\tau_1)$ \\
    Exact recovery, no band attempted & $(3+1/s)L$ & boundary bin; better bins committed \\
    Exact recovery after failed band & $(4+1/s)L$ & boundary bin; better bins committed \\
    Direct-exact route & $2L$ & boundary bin; better bins committed \\
    \bottomrule
  \end{tabular}
\end{table}

Before frontier-only refinement, the certified path contributes $R_{\mathrm{eq}}=1+1/s$,
replacing the second row-scale traversal identified in \cref{sec:bg:cuda-path} with a fractional
view.

Under the row-scaled traffic normalization of
\cref{sec:bg:interface}, let $p_{\mathrm{band}}$ be the probability of attempting the secondary band
and $p_{\mathrm{exact}}$ the probability of ultimately entering exact coarse recovery.  For the
persistent one-CTA sampled route,
\begin{equation}
  \mathbb{E}[R_{\mathrm{eq}}]
  =1+\frac{1}{s}+p_{\mathrm{band}}+2p_{\mathrm{exact}},
  \label{eq:method:expected-passes}
\end{equation}
while the frontier term $\sum_j|\mathcal Q_j|$ remains explicit in the complete work of
\cref{eq:method:fast-cost}.  Observed behavior on the target workload bears out this amortization:
primary underfill is rare, and the secondary boundary makes exact recovery rarer still.
\Cref{sec:eval:ablation} reports the measured $p_{\mathrm{band}}$ and
$p_{\mathrm{exact}}$.  Rows certified by the primary boundary incur neither the
secondary scan nor the exact-recovery passes.

KV-split exchanges additional aggregate sample work for a shorter validation span.  Each of the
$P$ segment CTAs reconstructs the same $L/s$ regular view, so the primary row-scaled traffic is
$R_{\mathrm{eq}}=1+P/s$; their disjoint Phase~2 segments still cover the row exactly once and
reduce the critical validation span from $L$ to approximately $L/P$.  For bounded $P$, the
per-row asymptotic work remains that of \cref{eq:method:fast-cost}, followed by a compact
histogram reduction and the same frontier work on the finisher.  The mapping therefore targets
occupancy at small row counts rather than reducing aggregate work.

Finally, the correctness argument is rank-general.  For a different $K$, the coarse target
$\widetilde K$, its sample rank, candidate budget, and the choice between sampled and
direct-exact routes must be recalibrated for efficiency; \cref{lem:method:containment}, the
full-row certificate, and the exact refinement recurrence do not change.  The implementation
constants in \cref{sec:impl:integration} specialize this general construction to the ranks and
row shapes of \cref{sec:bg:interface}. \Cref{sec:impl} also develops the persistent task
schedule, finisher election, storage lifetimes, and shape-stable dispatch that realize both
execution mappings.

\section{Implementation}
\label{sec:impl}

\subsection{Implementation Overview}
\label{sec:impl:overview}

HPC-Ops realizes the three phases of \cref{sec:method} as a family of exact GPU kernels.  Its
portfolio centers on the sampled long-row pipeline, realized through persistent and KV-split
execution, while row-local and eight-CTA direct-exact kernels cover complementary shape regimes.

These kernels differ in how row work and intermediate state are owned.  Persistent execution
keeps an entire row under one CTA; KV-split execution partitions its complete-row work and
combines compact segment state.  Row-local exact execution constructs exact coarse state within
one wider CTA, whereas the eight-CTA cluster constructs that state cooperatively.  Every path
then enters the same FP32 refinement interface, using a common ordered-key layout,
commit--retain recurrence, and output procedure.  A shape-stable planner selects the mapping
from captured dimensions.

The implementation discussion follows the sampled path through three CUDA concerns.
\Cref{sec:impl:localization} maps regular-view access and coarse localization onto registers and
reusable shared state; \cref{sec:impl:validation} develops the vectorized complete-row path, fused
candidate storage, and recovery traffic; and \cref{sec:impl:refinement} follows the compact
frontier's storage lifetime across exact digits.  \Cref{sec:impl:mappings} then specifies CTA
ownership, coordination, and dispatch across captured shapes.
\subsection{Phase 1: On-Chip Boundary Localization}
\label{sec:impl:localization}

The sampled kernels realize the row-phased regular view of \cref{sec:method:estimate} with
standard and wide-row policies that use different constant strides.  For row $r$, a CTA derives
offset $a_r$ from the row identity and streams the corresponding view $\mathcal V_{r,s}$ without
lookup state. This row-dependent phase distributes consecutive rows across sampling offsets
while preserving regular address generation within each row.  The resulting prefix locates the
primary boundary $\tau_0$ at the rank specified in \cref{sec:method:estimate}; the wide-row
policy also retains the deeper recovery boundary $\tau_1$ from the same prefix.  Under KV-split,
every segment CTA reconstructs the same row-local view from $a_r$, so boundary state need not be
produced by a designated CTA or broadcast across the group.

Each sampled score moves directly from global memory into a register, where it is projected once
to the 11-bit coarse key $c(x_i)$.  Shared-memory atomic increments accumulate these keys in the
2,048-bin histogram $h_{r,s}$, and an in-place CTA prefix produces $H_{r,s}$ and locates
$\tau_0$ and, when enabled, $\tau_1$.  A designated thread converts the selected bins to their
inclusive FP32 endpoints and publishes the resulting shared scalars.  The same histogram
allocation is then cleared for the leading exact-digit histogram in Phase~2.  The row remains in
global memory between the sampled accesses and the complete-row pass: registers hold the sampled
operands in flight, while shared memory retains only the histogram, prefix state, and endpoints.
The first column of \cref{fig:impl:dataflow} exposes this residency and reuse.

A view is usable only when it reaches its largest requested rank: $|\mathcal V_{r,s}|\ge q$ for a
primary-only policy and $|\mathcal V_{r,s}|\ge q_1$ when the recovery boundary is enabled.  Otherwise,
the CTA uses the same coarse-key representation, shared histogram, and in-place prefix to localize an
exact coarse boundary from all $L_r$ valid scores.  The sampled path passes one or two proposal
endpoints to full-row validation; the exact route passes the complete-row coarse boundary and its
prefix state to the classifier, which emits bins ahead of the boundary and materializes only its
boundary-bin frontier.  Both routes expose compact boundary state to the common FP32 refinement
interface.
\definecolor{HYPhaseOne}{HTML}{356FD3}
\definecolor{HYPhaseTwo}{HTML}{7457C7}
\definecolor{HYPhaseThree}{HTML}{16877F}
\definecolor{HYGlobalMem}{HTML}{3978B9}
\definecolor{HYGlobalMemLight}{HTML}{EDF5FC}
\definecolor{HYRegisterMem}{HTML}{D27927}
\definecolor{HYRegisterMemLight}{HTML}{FFF4E8}
\definecolor{HYSharedMem}{HTML}{2F8B5E}
\definecolor{HYSharedMemLight}{HTML}{EEF8F2}

\begin{figure}[!htbp]
  \centering
  \def\globalLineStrut{\vphantom{$\mathbf x=(x_0,\ldots,x_{L-1})$}}
  \def\registerLineStrut{\vphantom{$\kappa_{32}^{\downarrow}(x_i)$}}
  \begin{tikzpicture}[
    scale=0.96,
    transform shape,
    font=\footnotesize,
    scope/.style={rounded corners=2.5pt,draw=HYDarkBlue!46,
      fill=HYLightBlue!17,text=HYDarkBlue,align=center,
      inner xsep=4pt,inner ysep=2pt,line width=0.68pt,
      font=\fontsize{8.2}{8.9}\selectfont\bfseries},
    phase/.style={rounded corners=2.5pt,minimum height=0.64cm,
      align=center,inner sep=3pt,line width=0.76pt,
      text height=1.65ex,text depth=0.35ex,
      font=\fontsize{8.0}{8.8}\selectfont\bfseries},
    phaseOne/.style={phase,draw=HYPhaseOne!82,fill=HYPhaseOne!10,
      text=HYPhaseOne!92!black},
    phaseTwo/.style={phase,draw=HYPhaseTwo!82,fill=HYPhaseTwo!9,
      text=HYPhaseTwo!92!black},
    phaseThree/.style={phase,draw=HYPhaseThree!82,fill=HYPhaseThree!9,
      text=HYPhaseThree!92!black},
    lane/.style={rounded corners=3pt,line width=0.58pt},
    laneLabel/.style={rounded corners=2pt,align=center,inner sep=2pt,
      minimum width=1.18cm,font=\fontsize{7.0}{7.6}\selectfont\bfseries,
      line width=0.68pt},
    cell/.style={rounded corners=2.5pt,align=center,inner sep=3.2pt,
      line width=0.72pt,font=\fontsize{7.8}{8.6}\selectfont},
    gcell/.style={cell,draw=HYGlobalMem!82,fill=HYGlobalMemLight,
      text=HYGlobalMem!82!black},
    rcell/.style={cell,draw=HYRegisterMem!86,fill=HYRegisterMemLight,
      text=HYRegisterMem!86!black},
    scell/.style={cell,draw=HYSharedMem!84,fill=HYSharedMemLight,
      text=HYSharedMem!82!black},
    flow/.style={-{Latex[length=1.8mm]},line width=0.76pt,
      shorten <=1.2pt,shorten >=1.4pt},
    pOneFlow/.style={flow,draw=HYPhaseOne!82},
    pTwoFlow/.style={flow,draw=HYPhaseTwo!84},
    pThreeFlow/.style={flow,draw=HYPhaseThree!84},
    stateFlow/.style={-{Latex[length=1.7mm]},line width=0.72pt,
      draw=HYSharedMem!76!black,shorten <=1.3pt,shorten >=1.5pt},
    endpointFlow/.style={-{Latex[length=1.7mm]},dashed,line width=0.76pt,
      draw=HYRegisterMem!88!black,shorten <=1.4pt,shorten >=1.6pt},
    spillFlow/.style={-{Latex[length=1.7mm]},dashed,line width=0.72pt,
      draw=BrickRed!78,shorten <=1.4pt,shorten >=1.6pt},
    arrowText/.style={font=\fontsize{7.0}{7.6}\selectfont\ttfamily,
      inner sep=0.5pt,align=center},
    ledger/.style={rounded corners=2pt,minimum height=0.63cm,
      align=center,inner sep=2.5pt,line width=0.62pt,
      font=\fontsize{7.1}{7.8}\selectfont}
  ]
    \node[scope,text width=13.46cm,minimum height=0.44cm] at (8.15,7.18)
      {PERSISTENT SAMPLED EXECUTION};

    \node[phaseOne,text width=3.28cm] (ph1) at (3.20,6.45)
      {PHASE 1\enspace Boundary};
    \node[phaseTwo,text width=4.58cm] (ph2) at (7.70,6.45)
      {PHASE 2\enspace Validation};
    \node[phaseThree,text width=4.00cm] (ph3) at (12.66,6.45)
      {PHASE 3\enspace Refinement};

    \node[lane,draw=HYGlobalMem!25,fill=HYGlobalMemLight!55,
      minimum width=13.56cm,minimum height=1.17cm] at (8.15,5.26) {};
    \node[laneLabel,draw=HYGlobalMem!76,fill=HYGlobalMemLight,
      text=HYGlobalMem!84!black] at (0.66,5.26) {GLOBAL\\MEMORY};

    \node[lane,draw=HYRegisterMem!28,fill=HYRegisterMemLight!58,
      minimum width=13.56cm,minimum height=1.10cm] at (8.15,3.75) {};
    \node[laneLabel,draw=HYRegisterMem!80,fill=HYRegisterMemLight,
      text=HYRegisterMem!88!black] at (0.66,3.75) {REGISTERS};

    \node[lane,draw=HYSharedMem!27,fill=HYSharedMemLight!58,
      minimum width=13.56cm,minimum height=1.48cm] at (8.15,2.02) {};
    \node[laneLabel,draw=HYSharedMem!78,fill=HYSharedMemLight,
      text=HYSharedMem!86!black] at (0.66,2.02) {SHARED\\MEMORY};

    \node[gcell,text width=3.05cm,minimum height=0.96cm] (g1) at (3.20,5.26)
      {\globalLineStrut\textbf{FP32 score row}\\
       \globalLineStrut$\mathbf x=(x_0,\ldots,x_{L-1})$};
    \node[gcell,text width=4.25cm,minimum height=0.96cm] (g2) at (7.70,5.26)
      {\globalLineStrut\textbf{FP32 score row $\mathbf x$}\\
       \globalLineStrut candidate spill buffer};
    \node[gcell,text width=3.55cm,minimum height=0.96cm] (g3) at (12.66,5.26)
      {\globalLineStrut\textbf{FP32 score row $\mathbf x$}\\
       \globalLineStrut ping-pong spill; output $[0{:}K)$};

    \node[rcell,text width=3.05cm,minimum height=0.96cm] (r1) at (3.20,3.75)
      {\registerLineStrut sampled score $x_i$\\
       \registerLineStrut coarse key $c(x_i)$};
    \node[rcell,text width=4.25cm,minimum height=0.96cm] (r2) at (7.70,3.75)
      {\registerLineStrut current/next \texttt{float4}\\
       \registerLineStrut FP32 cutoff predicate};
    \node[rcell,text width=3.55cm,minimum height=0.96cm] (r3) at (12.66,3.75)
      {\registerLineStrut frontier item $(i,x_i)$\\
       \registerLineStrut select / drop / retain by digit};

    \node[scell,text width=3.05cm,minimum height=1.27cm] (s1) at (3.20,2.02)
      {coarse histogram $h_{r,s}$\\prefix $H_{r,s}$\\boundary bins $\tau_0$; optional $\tau_1$};
    \node[scell,text width=4.25cm,minimum height=1.27cm] (s2) at (7.70,2.02)
      {admitted indices $\mathcal A(\tau)$\\count $C(\tau)$; proposal $\widehat h_0$};
    \node[scell,text width=3.55cm,minimum height=1.27cm] (s3) at (12.66,2.02)
      {histogram $h_j\rightarrow h_{j+1}$\\frontier $\mathcal Q_j\rightarrow\mathcal Q_{j+1}$\\quota $k_j\rightarrow k_{j+1}$};

    \draw[pOneFlow] (g1.south) --
      node[midway,right=2pt,arrowText,text=HYPhaseOne!88!black]
      {strided sampled load} (r1.north);
    \draw[pOneFlow] (r1.south) --
      node[midway,right=2pt,arrowText,text=HYPhaseOne!88!black]
      {histogram \texttt{atomicAdd}} (s1.north);

    \draw[pTwoFlow] ([xshift=-5mm]g2.south) --
      node[midway,right=2pt,arrowText,text=HYPhaseTwo!90!black]
      {128-bit \texttt{ld.global.cg}} ([xshift=-5mm]r2.north);
    \draw[pTwoFlow] ([xshift=-5mm]r2.south) --
      node[midway,right=2pt,arrowText,text=HYPhaseTwo!90!black]
      {candidate append\\count/histogram \texttt{atomicAdd}}
      ([xshift=-5mm]s2.north);

    \draw[pThreeFlow] (g3.south) --
      node[midway,right=2pt,arrowText,text=HYPhaseThree!90!black]
      {gather $i\in\mathcal Q_j$} (r3.north);
    \draw[pThreeFlow] (r3.south) --
      node[midway,right=2pt,arrowText,text=HYPhaseThree!90!black]
      {frontier append\\histogram \texttt{atomicAdd}} (s3.north);

    \draw[endpointFlow] ([yshift=2mm]s1.east) to[out=24,in=198]
      node[pos=0.42,right=3pt,yshift=2.5pt,arrowText,text=HYRegisterMem!90!black]
      {FP32 cutoff} ([yshift=-2mm]r2.west);
    \draw[stateFlow] (s1.east) --
      node[midway,above=2pt,arrowText,text=HYSharedMem!82!black]
      {reuse} (s2.west);
    \draw[stateFlow] (s2.east) --
      node[midway,above=2pt,arrowText,text=HYSharedMem!82!black]
      {$\mathcal Q_0,h_0$} (s3.west);
    \draw[spillFlow] (s3.east)
      .. controls +(0.20,0.76) and +(0.20,-0.76) ..
      node[midway,right=1pt,arrowText,text=BrickRed!82!black]
      {overflow} (g3.east);

    \node[ledger,draw=HYPhaseOne!46,fill=HYPhaseOne!5,
      text=HYPhaseOne!88!black,text width=6.50cm] at (4.70,0.72)
      {\textbf{SCORE-ROW READS BY PHASE}\\
       {\color{black}\textbf{PHASE 1}}: $L/s$\hspace{0.65em}
       {\color{black}\textbf{PHASE 2}}: $L$\hspace{0.65em}
       {\color{black}\textbf{PHASE 3}}: $\sum_{j=0}^{D-1}|\mathcal Q_j|$};
    \node[ledger,draw=HYSharedMem!48,fill=HYSharedMem!5,
      text=HYSharedMem!84!black,text width=6.27cm] at (11.58,0.72)
      {\textbf{REUSED 2,048-BIN HISTOGRAM}\\
       same shared allocation: $h_{r,s}\rightarrow\widehat h_0/h_0\rightarrow h_1\rightarrow h_2$};
  \end{tikzpicture}
  \caption{GPU state residency and memory operations in persistent sampled execution.
    Columns trace boundary construction, full-row validation, and frontier refinement; rows separate
    global-memory, register, and shared-memory state.  Vertical arrows identify load, append, and
    atomic-update paths, whereas horizontal arrows pass the FP32 cutoff and certified frontier while
    recycling the shared histogram allocation.  Phase~2 combines 128-bit cache-global loads with
    candidate and leading-histogram construction.  Phase~3 gathers only unresolved frontier indices,
    with surviving overflow continuing through global ping-pong buffers.  The lower ledgers summarize
    per-phase score reads and the lifetime of the reused 2,048-bin histogram.}
  \label{fig:impl:dataflow}
\end{figure}
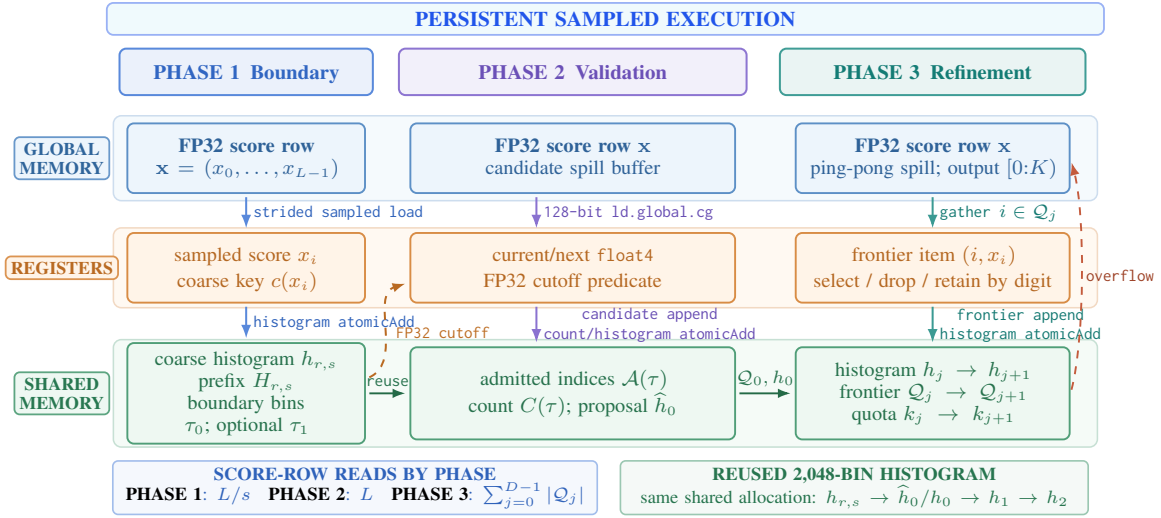

\FloatBarrier

\subsection{Phase 2: Fused Full-Row Validation}
\label{sec:impl:validation}

Phase~1 hands validation the inclusive FP32 endpoint represented by the selected coarse bin.
The CTA broadcasts its value through shared memory and retains both the value and its bit
pattern while streaming the row.  Comparing each original FP32 score directly with this endpoint
reproduces the selected coarse prefix, including its boundary value.  Coarse-key projection is
therefore confined to boundary construction and exact coarse recovery rather than repeated in
the bandwidth-dominant complete-row pass.

Aligned rows are fetched with 128-bit cache-global loads (\texttt{ld.global.cg.v4.f32});
two-wide and scalar paths cover less aligned layouts and the valid-row tail.  In the wide-row
single-CTA specialization, every thread keeps the current pair of \texttt{float4} vectors in
registers while issuing the next pair.  Classification and shared-memory updates consume the
current values as the following loads are in flight, overlapping memory latency and address
generation with the collective work.

The loop fuses endpoint classification, candidate formation, and construction of the leading
exact-digit histogram.  For every admitted score, an \texttt{atomicAdd} on the admitted count
reserves its candidate position; the same encounter stores the original row index and
contributes the score's leading ordered-key digit to the shared histogram.  Indices first occupy
the CTA-resident candidate array and continue into its global overflow slice when that array
fills.  Candidate storage, the realized count, and the leading histogram consequently describe
the same tail and proceed to refinement without a separate candidate or histogram pass.

After the traversal, the admitted count selects the continuation.  A sufficient primary tail reuses
the candidate workspace and histogram already produced by the fused loop.  On underfill, the
wide-row policy repeats the vectorized traversal but appends only the value band newly admitted by
$\tau_1$; continued underfill enters exact coarse localization.  The principal path therefore incurs
one complete-row traversal, while additional row traffic is confined to recovery.
\subsection{Phase 3: Frontier-Only Radix Refinement}
\label{sec:impl:refinement}

Phase~2 leaves the initial frontier $\mathcal Q_0$ and its leading exact-digit histogram $h_0$
in CTA-owned state.  The $11{+}11{+}10$ ordered-key partition lets every refinement round reuse
one 2,048-bin shared allocation.  A round prefixes the current histogram in place and identifies
the boundary digit.  Once that prefix has been consumed, the CTA clears the allocation and
repopulates it with the next histogram during classification.  The same storage thus advances
from $h_j$ to $h_{j+1}$ without rebuilding histogram state from the complete row.

Two shared arrays of 32-bit row indices physically represent the current and next frontiers.  In
round $j$, threads gather scores only for $i\in\mathcal Q_j$, form the required ordered-key
digit in registers, and route each index to the output, the discarded suffix, or the surviving
boundary group.  An atomic reservation appends a surviving index to the opposite array while
shared-memory increments construct its next-digit histogram; the two arrays then exchange roles.
Consequently, an FP32 score is gathered again only while its index remains unresolved.
Direct-exact mappings can instead retain compact $(\text{index},\text{ordered key})$ pairs on
chip and refine their resident keys without another score load.

Two global ping-pong slices extend the shared frontier when its resident capacity is exceeded.
Resident and spilled entries pass through the same digit classifier, and only surviving overflow is
written to the opposite slice.  Under persistent execution, each worker CTA owns these slices and
reuses them across the rows it later claims, so the spill workspace scales with the bounded worker
pool rather than the batch size.  Cooperative mappings receive segment- or row-owned slices from the
same precomputed workspace contract.  The global continuation therefore changes state residency,
not the refinement procedure.
\subsection{Execution Mappings and Dispatch}
\label{sec:impl:mappings}

HPC-Ops maps each captured shape to an execution organization that balances grid-level occupancy
against inter-CTA coordination.  Its sampled mappings cover long rows through either persistent
row ownership or partitioned validation, while row-local and eight-CTA direct-exact kernels
cover complementary shape regimes.  \Cref{fig:impl:portfolio} organizes these mappings by
captured row capacity $M$ and allocated row width $N$, and exposes their CTA ownership.  Within
the sampled long-row regime, abundant rows favor persistent execution, whereas lower row-level
concurrency favors KV-split.  For direct-exact execution, a grid of wider per-row CTAs avoids
peer coordination when it can sustain occupancy; only a few long rows instead favor an eight-CTA
cluster.  Inter-CTA cooperation is therefore introduced where the work per captured row can
amortize it.

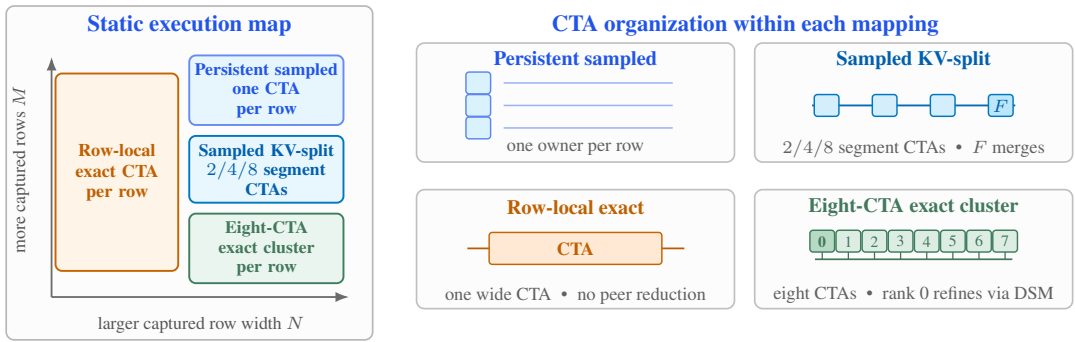
\begin{figure}[!htbp]
  \centering
  \begin{tikzpicture}[
    scale=0.96,
    transform shape,
    font=\footnotesize,
    map/.style={rounded corners=3pt,draw=black!28,fill=black!1,
      line width=0.65pt},
    region/.style={rounded corners=2.5pt,align=center,inner sep=3pt,
      font=\scriptsize\bfseries,line width=0.72pt},
    persistent/.style={region,draw=HYDarkBlue!68,fill=HYLightBlue!28,
      text=HYDarkBlue},
    kvsplit/.style={region,draw=cyan!55!blue,fill=cyan!8,
      text=cyan!35!blue},
    rowlocal/.style={region,draw=orange!78!black,fill=orange!9,
      text=orange!75!black},
    cluster/.style={region,draw=ForestGreen!64!black,fill=ForestGreen!7,
      text=ForestGreen!58!black},
    card/.style={rounded corners=3pt,draw=black!25,fill=black!1,
      minimum height=1.64cm,inner sep=4pt,line width=0.62pt},
    cta/.style={rounded corners=1.2pt,minimum width=0.34cm,
      minimum height=0.30cm,inner sep=1pt,line width=0.55pt},
    pcta/.style={cta,draw=HYDarkBlue!72,fill=HYLightBlue!42},
    kcta/.style={cta,draw=cyan!55!blue,fill=cyan!15},
    rcta/.style={cta,draw=orange!78!black,fill=orange!17},
    ccta/.style={cta,draw=ForestGreen!64!black,fill=ForestGreen!12}
  ]
    \node[map,minimum width=5.05cm,minimum height=4.60cm] (mapbox) at (2.72,2.62) {};
    \node[font=\footnotesize\bfseries,text=HYDarkBlue] at (2.72,4.65)
      {Static execution map};
    \node[font=\footnotesize\bfseries,text=HYDarkBlue] at (10.38,4.65)
      {CTA organization within each mapping};
    \draw[-{Latex[length=1.8mm]},line width=0.68pt,draw=black!58]
      (0.82,0.91) -- (4.92,0.91);
    \draw[-{Latex[length=1.8mm]},line width=0.68pt,draw=black!58]
      (0.82,0.91) -- (0.82,4.22);
    \node[font=\scriptsize,text=black!66] at (2.88,0.53)
      {larger captured row width $N$};
    \node[font=\scriptsize,text=black!66,rotate=90] at (0.40,2.67)
      {more captured rows $M$};
    \node[rowlocal,text width=1.48cm,minimum height=2.72cm]
      at (1.72,2.63) {Row-local\\exact CTA\\per row};
    \node[persistent,text width=2.02cm,minimum height=0.92cm,inner xsep=1.5pt,
      font=\fontsize{6.8}{7.6}\selectfont\bfseries]
      at (3.78,3.76) {Persistent sampled\\one CTA\\per row};
    \node[kvsplit,text width=2.02cm,minimum height=0.92cm,inner xsep=1.5pt,
      font=\fontsize{6.8}{7.6}\selectfont\bfseries]
      at (3.78,2.67) {Sampled KV-split\\$2/4/8$ segment\\CTAs};
    \node[cluster,text width=2.02cm,minimum height=0.92cm,inner xsep=1.5pt,
      font=\fontsize{6.8}{7.6}\selectfont\bfseries]
      at (3.78,1.58) {Eight-CTA\\exact cluster\\per row};

    \node[card,text width=4.10cm] (pcard) at (8.04,3.59) {};
    \node[font=\fontsize{8.0}{9.0}\selectfont\bfseries,text=HYDarkBlue]
      at (8.04,4.17) {Persistent sampled};
    \foreach \y in {3.24,3.55,3.86} {
      \draw[line width=0.58pt,draw=HYDarkBlue!45] (7.05,\y) -- (9.38,\y);
      \node[pcta] at (6.70,\y) {};
    }
    \node[font=\fontsize{7.2}{8.0}\selectfont,text=black!62] at (8.04,2.95)
      {one owner per row};

    \node[card,text width=4.10cm] (kcard) at (12.70,3.59) {};
    \node[font=\fontsize{8.0}{9.0}\selectfont\bfseries,text=cyan!35!blue]
      at (12.70,4.17) {Sampled KV-split};
    \draw[line width=0.70pt,draw=cyan!45!blue] (11.30,3.55) -- (14.10,3.55);
    \foreach \x in {11.50,12.30,13.10} {
      \node[kcta] at (\x,3.55) {};
    }
    \node[kcta,fill=cyan!24,font=\fontsize{6.6}{7.0}\selectfont\bfseries,
      text=cyan!35!blue] at (13.90,3.55) {$F$};
    \node[font=\fontsize{7.2}{8.0}\selectfont,text=black!62] at (12.70,2.95)
      {$2/4/8$ segment CTAs \;\textbullet\; $F$ merges};

    \node[card,text width=4.10cm] (rcard) at (8.04,1.57) {};
    \node[font=\fontsize{8.0}{9.0}\selectfont\bfseries,text=orange!75!black]
      at (8.04,2.15) {Row-local exact};
    \draw[line width=0.66pt,draw=orange!72!black] (6.55,1.58) -- (9.54,1.58);
    \node[rcta,minimum width=2.38cm,minimum height=0.44cm,
      font=\fontsize{7.2}{8.0}\selectfont\bfseries,text=orange!75!black]
      at (8.04,1.58) {CTA};
    \node[font=\fontsize{7.2}{8.0}\selectfont,text=black!62] at (8.04,0.93)
      {one wide CTA \;\textbullet\; no peer reduction};

    \node[card,text width=4.10cm] (ccard) at (12.70,1.57) {};
    \node[font=\fontsize{8.0}{9.0}\selectfont\bfseries,text=ForestGreen!58!black]
      at (12.70,2.15) {Eight-CTA exact cluster};
    \node[ccta,fill=ForestGreen!28,font=\fontsize{6.2}{6.8}\selectfont\bfseries,
      text=ForestGreen!58!black] at (11.44,1.68) {0};
    \foreach \x/\lab in {11.80/1,12.16/2,12.52/3,12.88/4,13.24/5,13.60/6,13.96/7} {
      \node[ccta,font=\fontsize{6.2}{6.8}\selectfont,text=ForestGreen!58!black]
        at (\x,1.68) {\lab};
    }
    \foreach \x in {11.44,11.80,12.16,12.52,12.88,13.24,13.60,13.96} {
      \draw[line width=0.48pt,draw=ForestGreen!48!black] (\x,1.53) -- (\x,1.43);
    }
    \draw[line width=0.64pt,draw=ForestGreen!55!black] (11.34,1.43) -- (14.06,1.43);
    \node[font=\fontsize{7.2}{8.0}\selectfont,text=black!62] at (12.70,0.94)
      {eight CTAs \;\textbullet\; rank 0 refines via DSM};
  \end{tikzpicture}
  \caption{Static execution mapping and CTA organization.  The left panel maps captured row
    capacity $M$ and allocated row width $N$ to the primary mapping; the qualitative regions also
    depend on device capability.  The right cards show one owner per row for persistent sampled execution,
    segment CTAs and a last-arriving finisher $F$ for KV-split, one wide CTA for row-local exact
    execution, and eight clustered CTAs whose rank 0 refines DSM-shared state.}
  \label{fig:impl:portfolio}
\end{figure}

\FloatBarrier

\paragraph{Persistent sampled.}
When the call exposes enough independent rows to occupy the device, one 512-thread CTA owns the
complete state of a row.  Keeping the row within one CTA removes inter-CTA coordination from the
common sampled path.  The first scheduling wave assigns one row directly to each CTA; CTAs that
finish this wave claim remaining rows from a device counter.  Direct first-wave assignment
avoids queue traffic for modest calls, while completion-order claiming balances large ragged
batches.  The sampled histogram, FP32 endpoint, candidate frontier, and exact histograms remain
CTA-owned across all three stages; frontier overflow changes residency without introducing a
peer owner.  The stride-64 and stride-128 views are compile-time policies of this common body.

\paragraph{Sampled KV-split.}
For sampled long-row shapes with limited row-level concurrency, the complete-row scan dominates
while a one-CTA-per-row grid leaves the device underoccupied.  Two, four, or eight CTAs
therefore process disjoint contiguous row segments.  Each CTA reconstructs the same row-phased
regular view and hence the same sampled endpoint proposal; this small replicated view is
amortized by the much longer segment scan.  During the complete-row stage, a CTA writes a
private candidate range, candidate count, and leading exact-digit histogram for its segment.
The aggregate validation scan still covers every valid score exactly once, while its parallel
span falls with the number of participating CTAs.

Each segment publishes its state before incrementing a row-local arrival counter.  The
last-arriving CTA becomes the finisher, reduces the fixed-width histograms, compacts the segment
candidate ranges, and evaluates the aggregate count.  A sufficient aggregate enters frontier
refinement on the finisher; an underfilled aggregate enters the exact coarse route.  The split
count is selected as part of the launch geometry, while the finisher exposes the same logical
row-level selection state to refinement as the one-CTA mapping.

\paragraph{Row-local exact.}
When one wide block provides sufficient per-row parallelism, row-local ownership avoids the
coordination cost of a cooperative mapping.  A 1024-thread CTA constructs the exact coarse
boundary and completes selection with row-local state.  The short-row specialization keeps
several aligned score vectors in registers across boundary construction and classification,
carrying the loaded values into the exact path without materializing a row-sized candidate
workspace.  The longer-row specialization streams scores under the same row-local ownership,
using either direct row assignment or a bounded worker schedule according to the captured row
count.  Both variants retain compact exact boundary state on chip and reuse the common
three-digit FP32 refinement.

\paragraph{Eight-CTA exact cluster.}
When only a few long rows are available, even a wide row-local CTA leaves insufficient blocks to
fill the device.  An eight-CTA hardware cluster converts row length into intra-row parallelism
by assigning one contiguous segment to each CTA.  The CTAs first construct disjoint exact coarse
histograms and combine them through distributed shared memory.  Once the common coarse boundary
is known, each CTA rescans its segment and publishes the definitely selected prefix and
boundary-bin candidates.  Rank 0 owns FP32 frontier refinement and directly traverses the
compact DSM-resident state, while peer CTAs publish disjoint selected ranges and keep their
distributed state live until row output is complete.  The mapping turns coarse localization and
boundary classification into intra-row parallel work, then leaves only the reduced frontier
under one owner for fine-grained refinement.

\paragraph{Shape-stable dispatch.}
\label{sec:impl:integration}
One host planner realizes this execution map before graph capture from the captured rank, row
capacity, padded width, row pitch, and device capability.  It fixes the kernel family, view
density, KV-split count, and workspace contract behind a single entry point; the same plan
drives workspace sizing, while a persistent direct-exact body supplies a shape-compatible
fallback when sampling is not selected or a cooperative launch cannot be formed.  The live row
count and per-row valid lengths remain device inputs, so capture and replay preserve launch
geometry, queue counters, split state, and spill ownership without allocation, host readback, or
data-dependent relaunch; persistent counters are restored on device.

\clearpage
\section{Evaluation}
\label{sec:eval}

We evaluate exact Top-$K$ selection from standalone operator scaling to framework-derived
long-context traces and score-and-select integration.  We first measure how performance changes
with rows per call and allocated row width, then quantify whether the gains persist under native
scheduling and score-matrix workspace partitioning.  Policy and mechanism ablations finally
identify the source of the gain and the row shapes for which sampled localization is effective.

\subsection{Experimental Setup}
\label{sec:eval:setup}

\paragraph{Experimental environment.}
We evaluate HPC-Ops Top-K on NVIDIA H20 GPUs using FP32 indexer scores produced by
Hy4-Preview~\cite{hy4preview2026} and captured immediately before selection.  The evaluated
standalone operator uses $K=2048$ and preserves the native ragged row lengths.
\Cref{tab:eval:env} summarizes the hardware and software environment.

\begin{table}[!htbp]
\centering
\captionsetup{skip=6pt}
\caption{Evaluation environment.  CUTLASS Python DSL is used by the
  TensorRT-LLM CuTe kernels.}
\label{tab:eval:env}
\begin{tabular}{ll}
\toprule
Component & Configuration \\
\midrule
GPU & \texttt{NVIDIA H20} (78 SMs, \SI{96}{GiB}) \\
Selection & Top-$K$ \\
Score tensor & FP32 \\
CUDA runtime & \texttt{12.8} \\
PyTorch & \texttt{2.11.0+cu128} \\
CUTLASS Python DSL & \texttt{4.5.0} \\
\bottomrule
\end{tabular}

\end{table}

\paragraph{Timing protocol.}
We measure GPU device time with CUDA events on the execution stream.  After the
workload-specific warmup and conditioning schedule, repeated CUDA-Graph measurements produce the
per-call median.  The timed region contains operator execution; input preparation, allocation,
reference construction, correctness checks, and one-time compilation are excluded.  The
terminal-step parameter and complete-step mechanism controls use 10 warmups, two conditioning
sweeps of 50 replays, and seven rotated rounds that each average 100 replays.  The crossover
sweep uses 10 warmups and seven rotated 100-replay rounds, while the score-and-select
integration check uses eight warmups, two 25-replay conditioning sweeps, and seven 50-replay
rounds.  Trace- and step-level latency is the sum of the corresponding per-call medians.

\paragraph{Exactness criterion.}
\label{sec:eval:correctness}
For the operator, framework, and terminal-step experiments, we validate every
returned row under the set-valued exactness contract.  Its indices must be
distinct and lie in the valid prefix $[L_r]$.  For $L_r>K$, a set $S_r$ is
exact when $|S_r|=K$ and
\begin{equation}
  \min_{i\in S_r} x_{r,i}
  \;\ge\;
  \max_{j\in [L_r]\setminus S_r} x_{r,j},
  \label{eq:eval:exactness}
\end{equation}
with arbitrary choices permitted among values tied at the boundary.  When
$L_r\le K$, the result must contain every valid index and use the contract
sentinel for the remaining output slots.  The criterion accepts any valid
choice among boundary ties without imposing output order or a floating-point
tolerance.  Latency ratios use exact executions; invalid results, illegal
accesses, and out-of-memory outcomes are reported separately.
The score-and-select integration applies the same criterion to verified rows
from every constituent call.

\paragraph{Implementations.}
We evaluate the production entry of HPC-Ops and the closest available exact
implementations from five external families:
\begin{itemize}[leftmargin=1.5em,itemsep=0.18em,topsep=0.3em]
  \item \textbf{HPC-Ops.} We invoke \texttt{hpc.topk} with its
    production workspace and internal shape dispatch.  We use commit
    \texttt{f39028d}.
  \item \textbf{vLLM.} We evaluate two upstream sampler entries~\cite{vllm2023}.  Prefill uses
    \texttt{top\_k\_per\_row\_}\allowbreak\texttt{prefill}; decode uses
    \texttt{top\_k\_per\_row\_}\allowbreak\texttt{decode} with \texttt{next\_n=1}.  We use commit
    \texttt{c5d840f}.
  \item \textbf{TensorRT-LLM.} We evaluate
    \texttt{cute\_dsl\_topk\_prefill\_}\allowbreak\texttt{wrapper} with both \textsc{reread} and
    \textsc{gmem-spill} overflow policies~\cite{tensorrtllm2026}.  Its decode path invokes
    \texttt{cute\_dsl\_radix\_filter\_topk\_}\allowbreak\texttt{wrapper}; the eight-CTA path invokes
    \texttt{cute\_dsl\_radix\_filter\_topk\_}\allowbreak
    \texttt{single\_pass\_multi\_cta\_wrapper}.  We use commit
    \texttt{f5cbe6b}.
  \item \textbf{SGLang.} We evaluate the JIT Top-K V2 path~\cite{sglang2024}
    \texttt{topk\_transform\_512\_v2}, with metadata produced by
    \texttt{plan\_topk\_v2}.  We use commit \texttt{155aa26}.
  \item \textbf{FlashInfer.} We evaluate \texttt{flashinfer.top\_k}~\cite{flashinfer2025} with
    automatic algorithm selection and \texttt{sorted=False}.  We use commit
    \texttt{56ed540}.
  \item \textbf{PyTorch.} We evaluate \texttt{torch.topk}~\cite{pytorch2019} with
    \texttt{largest=True} and \texttt{sorted=False} as the generic dense
    selector.  We use commit \texttt{70d99e9}.
\end{itemize}
All implementations operate on the same valid FP32 score prefixes.  Native
variable-length interfaces consume the original row lengths directly.  For
dense-only interfaces, a full-suffix-mask adapter sets every invalid suffix
position to negative infinity before timing while preserving the original row
count, allocated width, and invocation boundaries.

\subsection{Operator-Level \texorpdfstring{Top-$K$}{Top-K} Performance}
\label{sec:eval:scaling}

We first isolate the two dimensions that determine the dominant score traffic: the rows per call
and their allocated width.  The matrix combines $M\in\{512,1024,2048,4096,8192\}$ with four
terminal widths of 128K, 256K, 512K, and approximately 1M elements.  All inputs are unchanged
scores from the Hy4-Preview captures.  For each width $N$, we retain its final 8192 consecutive
causal rows, so their valid prefixes span the final 8192 context lengths ending at $N$.
Consecutive non-overlapping groups form the smaller values of $M$, so the sweep changes launch
occupancy and aggregate work without replacing or rescaling the score distribution.

Every implementation receives the same full-width FP32 tensor for a given shape.  Invalid causal
suffixes are set to negative infinity before timing; length-aware implementations additionally
receive the original row lengths. The latency of a shape is the median over all groups and
repeated executions. For each external implementation family, we retain its fastest entry point
that passes \cref{eq:eval:exactness} at that shape, and compare HPC-Ops with the fastest of
these family-level results.

\begin{figure}[!htbp]
  \centering
  \includegraphics[width=\linewidth]{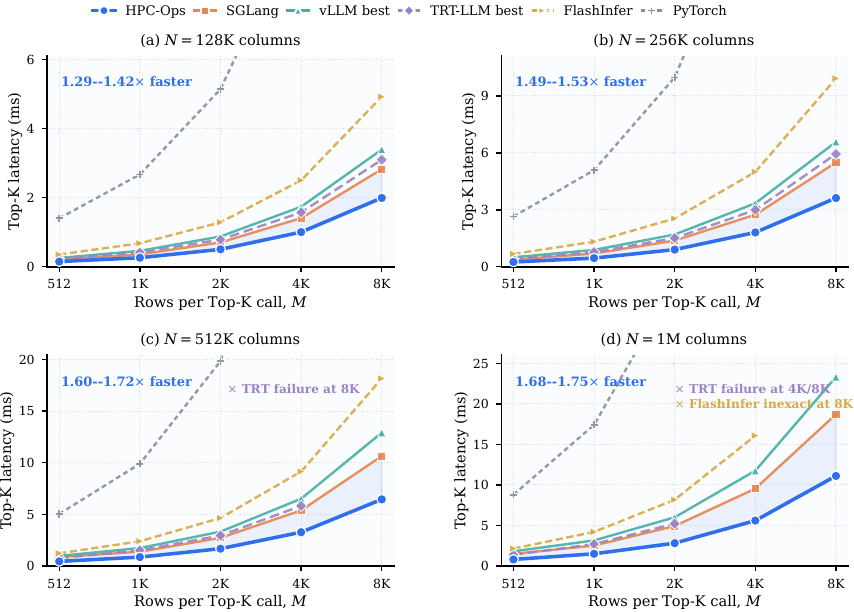}
  \caption{Exact Top-$K$ latency on captured causal terminal windows.  Each panel
    fixes the allocated row width $N$ and sweeps the rows per call $M$.  Every
    implementation receives the same full-width tensor with invalid causal
    suffixes masked to negative infinity.  Baseline curves retain the fastest
    verified exact entry point in each family at every shape; missing points
    denote no verified exact result.  PyTorch is retained as a masked-dense
    reference, with values beyond the shared panel ranges clipped to preserve
    resolution among the competitive kernels.  HPC-Ops is fastest in all 20
    shapes by $1.29$--$1.75\times$ ($1.55\times$ geometric mean).}
  \label{fig:eval:scaling}
\end{figure}

HPC-Ops passes exact verification and is fastest in all 20 configurations. Its advantage over
the best exact external result is $1.29$--$1.42\times$ at 128K columns, $1.49$--$1.53\times$ at
256K, $1.60$--$1.72\times$ at 512K, and $1.68$--$1.75\times$ at 1M, for a geometric mean of
$1.55\times$ over the matrix.  SGLang supplies the best-exact result in 18 of the 20 shapes;
TensorRT-LLM decode supplies the other two at $M=512$ on the longest widths. The gain is
comparatively stable across $M$ at a fixed width, but widens as $N$ grows.  This trend matches
the intended long-row regime: increasing the row count mainly changes parallel occupancy,
whereas increasing the context length amplifies the full-row traffic avoided during boundary
localization.

Latency is not the only scaling limit at million-token widths.  TensorRT-LLM's GMEM-spill policy
handles candidates that exceed on-chip capacity by reserving $2MN$ additional 32-bit entries.
Because this reservation follows the full $M\times N$ score shape rather than the realized
overflow volume, its memory cost grows in lockstep with the input and eventually causes an
out-of-memory failure.  HPC-Ops does not rely on a separate matrix-wide spill allocation.
Avoiding this second score-matrix-scale footprint leaves substantially more device memory
available as the row batch and context length grow together.

\subsection{Framework-Level \texorpdfstring{Top-$K$}{Top-K} Performance}
\label{sec:eval:framework}

Framework scheduling determines the sequence of row groups and prefix widths presented to
Top-$K$ during a long-context request.  In the vLLM-based execution considered here, query work
from prefill and decode is first organized into bounded scheduler steps.  The indexer then
partitions each step again to keep its materialized FP32 score matrix within a fixed workspace
budget, so one scheduler step may issue several Top-$K$ invocations.  As the prefix width $N$
grows, the number of rows that fit in one invocation decreases approximately in inverse
proportion to $N$.  The resulting trace therefore moves toward wider rows and smaller row groups
over the course of a request.

The framework replay caps each scheduler step at 64K query tokens. Under this policy, the 245K
and 430K requests span four and seven scheduler steps but generate 26 and 71 Top-$K$
invocations, respectively. We preserve each invocation's original $M$, allocated $N$, and
per-row valid lengths $L_r$, together with the original invocation boundaries and order. Every
implementation therefore sees the same framework-derived shape sequence.  The timing protocol of
\cref{sec:eval:setup} is applied throughout each trace, and we report its trace-level Top-$K$
latency.

\begin{figure}[!htbp]
  \centering
  \includegraphics[width=\linewidth]{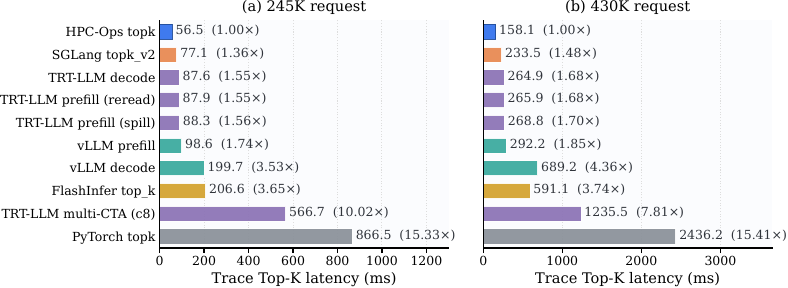}
  \caption{Exact Top-$K$ latency over the complete framework traces of the
    245K and 430K requests.  Parenthesized values report latency relative to
    HPC-Ops.  Every implementation receives the same framework-derived
    sequence of row shapes and valid lengths.}
  \label{fig:eval:framework}
\end{figure}

HPC-Ops remains exact and fastest for both requests.  Compared with SGLang, the fastest external
exact baseline, HPC-Ops reduces trace-level Top-$K$ latency from 77.15 to 56.53 ms at 245K and
from 233.50 to 158.12 ms at 430K, yielding speedups of $1.36\times$ and $1.48\times$,
respectively.  The vLLM and TensorRT-LLM paths also remain exact but require more time over both
traces, while the exact masked dense paths are slower still (\cref{fig:eval:framework}).

The increase from $1.36\times$ at 245K to $1.48\times$ at 430K follows the progression toward
wider rows described above.  As a request grows, later chunks account for a larger share of
trace-level Top-$K$ latency and place more of the execution in the wide-row regime, where
sampled boundary localization avoids more full-row traffic.  Their growing weight strengthens
the trace-level benefit, reaching $1.55\times$ at the terminal 430K prefix.

\paragraph{Score-and-select integration.}
We additionally place the selector behind the same DeepGEMM~\cite{deepgemm2025} FP8 MQA score
producer and capture both operations in one CUDA graph.  Replacing SGLang V2 with HPC-Ops
reduces the Top-$K$ segment from 27.41 to 18.15 ms at 245K and from 38.01 to 24.06 ms at 430K,
corresponding to $1.51\times$ and $1.58\times$.  The complete score-and-select chains fall from
339.99 to 330.84 ms and from 488.01 to 474.70 ms, saving 9.15 and 13.31 ms.  These measurements
cover the terminal score-production and selection chain; model layers outside that chain are not
included.

\subsection{Sampled-Path Ablations}
\label{sec:eval:ablation}

Having established the gains at both the operator and framework levels, we next examine how the
sampled path is configured, protected against underfill, and dispatched across row shapes.  We
first vary the retention margin $\rho$ and view stride $s$ to select a compact primary proposal,
then evaluate how the secondary target absorbs the remaining underfills.  A matched
exact-proposal control and an $(M,N)$ sweep finally identify where the gain comes from and over
which row shapes sampling is worthwhile.

The margin, stride, and rescue experiments replay the final scheduler steps of the 245K and 430K
traces at their native invocation boundaries.  They preserve every call's original $M$, $N$,
$L_r$, and execution order, and report step latency as the sum of per-call medians.  Recovery
rates are aggregated over all rows in the corresponding step, and every reported variant
satisfies \cref{eq:eval:exactness}.

\paragraph{Retention margin.}
We first fix the view stride at $s=128$ and sweep the requested complete-row rank $\widetilde
K=\rho K$ from $K$ to $2K$.  Moving the proposal deeper into the upper tail gives sampling error
more room before underfill, but also retains a larger candidate set for exact refinement.

\begin{figure}[!htbp]
  \centering
  \includegraphics[width=\linewidth]{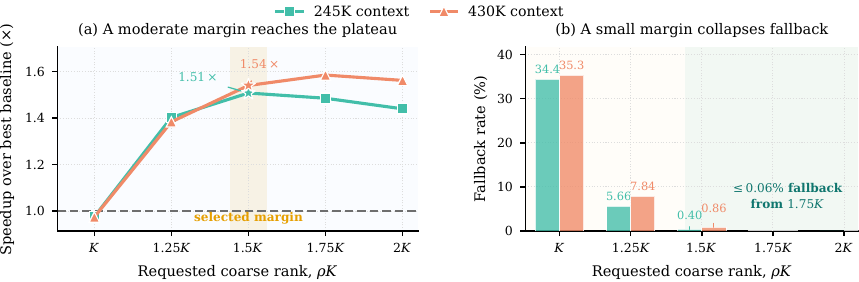}
  \caption{Retention-margin ablation on the terminal framework chunks.
    With $s=128$ and the secondary boundary disabled, panel (a) reports exact
    speedup over the fastest external baseline and panel (b) reports the rows
    entering exact recovery.  A target of $\widetilde K=1.5K$ reaches the
    shared latency knee while keeping primary recovery below 1\%.}
  \label{fig:eval:retention-margin}
\end{figure}

At $\rho=1$, the proposal has no allowance for estimation error, sending roughly one third of
the rows to exact recovery and erasing the sampled path's advantage.  Increasing the target to
$1.25K$ lowers recovery to 5.7--7.8\%, while $1.5K$ lowers it below 1\% and reaches
$1.51$--$1.54\times$ speedup (\cref{fig:eval:retention-margin}).  Deeper targets suppress the
remaining recoveries but retain more candidates on every row, producing a broad latency plateau
rather than a sharper optimum.  The $1.5K$ target therefore supplies a compact tail with
sufficient calibration margin on both traces.

\paragraph{Sampling stride.}
We next hold the requested complete-row rank at $\widetilde K=1.5K$ and vary the view stride $s$
from 16 to 256 under the same primary-only policy. Increasing $s$ reduces the sampled proposal
work in proportion to the view density, but a view that becomes too sparse incurs more exact
recovery.

\begin{figure}[!htbp]
  \centering
  \includegraphics[width=\linewidth]{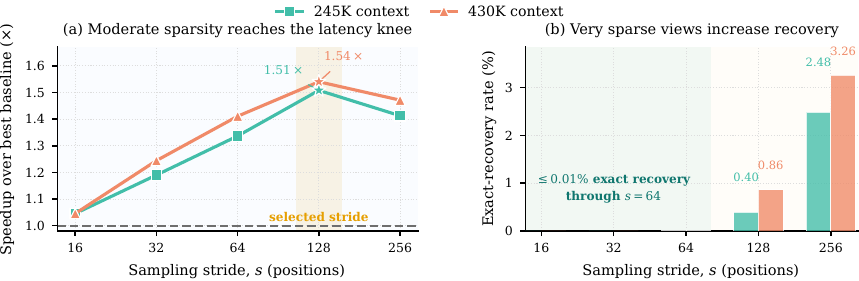}
  \caption{Sampling-stride ablation at $\widetilde K=1.5K$ on the same terminal
    framework chunks.  Panel (a) reports exact speedup over the fastest external
    baseline, and panel (b) reports primary exact recovery.  Moderate sparsity
    reduces proposal work through $s=128$; the higher recovery rate at $s=256$
    reverses part of the gain.}
  \label{fig:eval:sampling-stride}
\end{figure}
\FloatBarrier

The denser $s=16$ view is almost recovery-free but yields only about $1.05\times$ speedup
because proposal construction still reads a larger fraction of each row.  Recovery remains below
0.01\% through $s=64$, and $s=128$ raises the speedup to $1.51$--$1.54\times$ while keeping
recovery below 0.9\%.  At $s=256$, recovery rises to 2.48--3.26\% and the speedup falls on both
traces (\cref{fig:eval:sampling-stride}).  The common knee at $s=128$ balances the shrinking
proposal view against the growing cost of exact recovery.

\paragraph{Conditional rescue.}
Figures~\ref{fig:eval:retention-margin} and \ref{fig:eval:sampling-stride} disable the secondary
boundary to expose the primary proposal in isolation.  The default policy instead records a
secondary complete-row target $\widetilde K_1=1.75K$ alongside the primary target $\widetilde
K=1.5K$.  The same sampled histogram produces the row-local boundaries $\tau_{0,r}$ and
$\tau_{1,r}$ for these targets.  Only a row whose primary count satisfies $C(\tau_{0,r})<K$
scans the intervening band up to $\tau_{1,r}$; all other rows enter refinement with the original
compact frontier.  We compare this conditional policy with a primary-only variant that sends the
same underfilled rows directly to exact coarse recovery.

\begin{figure}[!htbp]
  \centering
  \includegraphics[width=\linewidth]{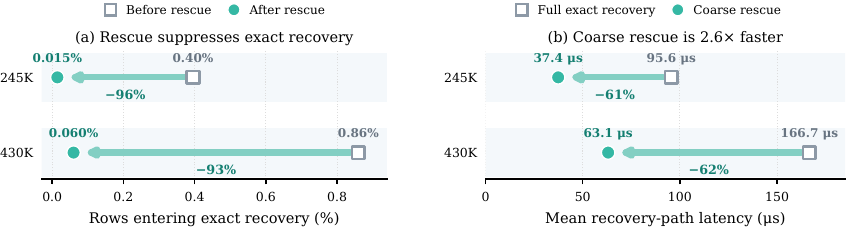}
  \caption{Conditional-rescue ablation at $s=128$.  Panel (a) compares the
    fraction of rows entering
    exact recovery before and after applying the secondary boundary.  Panel
    (b) compares the mean recovery-path latency of coarse rescue with full
    exact recovery on affected rows.  Conditional rescue resolves 93--96\% of
    primary misses, while coarse rescue requires 61--62\% less recovery-path
    time than full exact recovery.}
  \label{fig:eval:conditional-rescue}
\end{figure}

The primary boundary underfills 0.40\% and 0.86\% of rows in the two terminal steps.
Conditional rescue reduces the residual exact-recovery rates to 0.015\% and 0.060\%, resolving
96\% and 93\% of those misses, respectively (\cref{fig:eval:conditional-rescue}).  On the
affected rows, coarse rescue reduces mean recovery-path latency from 95.6 to 37.4~$\mu$s at 245K
and from 166.7 to 63.1~$\mu$s at 430K.  The secondary boundary therefore protects the compact
primary frontier against rare underestimates without adding a common-path scan.

\paragraph{Mechanism attribution and activation regime.}
We compare exact and sampled boundary localization at the same complete-row target $\widetilde
K=1.5K$ while holding the downstream exact selector fixed.  The exact control derives its
boundary from the complete row, whereas the sampled path uses an $s=128$ view.  We then sweep
the allocated row width $N$ from 16K to 64K and evaluate seven row counts $M$ from 64 to 4096
over two captured score distributions, yielding 14 workload cells at each width. Together, these
controls identify both the source of the common-path gain and the row shapes over which sampling
amortizes its fixed overhead.

\begin{figure}[!htbp]
  \centering
  \includegraphics[width=\linewidth]{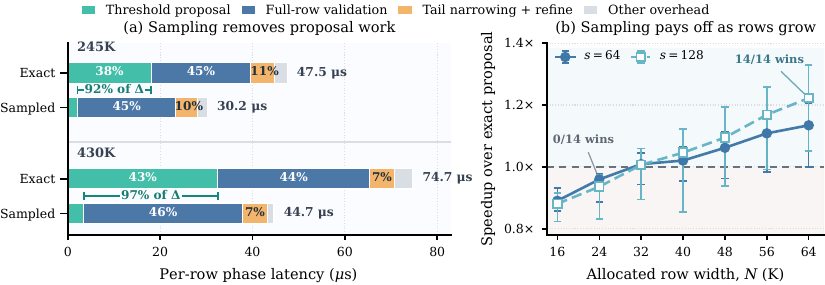}
  \caption{Mechanism attribution and sampling crossover.  Panel (a)
    decomposes additive $M=1$ active-CTA phase latency for matched exact- and
    sampled-proposal controls; brackets report the fraction of their difference
    explained by boundary localization.  Panel (b) reports the median speedup
    over exact proposal across 14 workload cells---seven $M$ values over two
    captures---at each width, with error bars spanning the observed minimum and
    maximum.  The $s=128$ series uses the default dual-boundary policy.}
  \label{fig:eval:mechanism}
\end{figure}

In the additive per-row measurements, sparse localization reduces proposal time from
approximately 18 to 2~$\mu$s at 245K and from 33 to 3.5~$\mu$s at 430K.  The reduction accounts
for 92\% and 97\% of the latency difference, while full-row validation and exact tail processing
remain comparable across the matched controls (\cref{fig:eval:mechanism}).  Over the complete
terminal steps, replacing the exact proposal with the sampled proposal lowers latency from 32.8
to 18.3~ms and from 43.2 to 24.2~ms, corresponding to $1.79\times$ and $1.78\times$.  Together,
these results identify boundary localization as the dominant source of the complete-step gain.

The width sweep exposes a transition rather than a single universal cutoff. At 16K and 24K,
exact proposal wins all tested cells.  Results become mixed from 32K through 56K as row count
and capture geometry determine whether the saved proposal work amortizes sampling.  At 64K, both
sampled variants win all 14 cells; their median speedups reach $1.13\times$ for $s=64$ and
$1.22\times$ for the default $s=128$ policy.  This regime supports dispatching shorter rows to
exact localization while reserving sampling for sufficiently wide rows, with the transition
handled as a range rather than a hard threshold.

\FloatBarrier

\section{Related Work}
\label{sec:related}

\Cref{sec:bg:landscape} compares algorithms that accept the same materialized score rows and
return exact Top-$K$ indices.  This section considers systems that reduce the surrounding
sparse-attention cost by moving selection into score production, sharing indexing work across
queries or layers, or relocating the data and execution substrate.  These approaches change the
standalone boundary defined in \cref{sec:bg:interface}, although several still require token
selection within their resulting pipelines.

\paragraph{Fused and streaming indexer selection.}
Fusing selection into score production can eliminate the complete materialized score matrix.
LiteTopK integrates exact Top-$K$ into the indexer and filters candidates as their scores are
produced~\cite{litetopk2026}; StreamIndex evaluates the indexer in chunks and merges bounded
candidate sets to control the intermediate memory of compressed sparse
attention~\cite{streamindex2026}. Both consume query and key representations rather than an
already materialized $\mathbf X\in\mathbb R^{M\times N}$.  Their optimization boundary is
therefore the fused indexer pipeline rather than the standalone materialized-row operator.

\paragraph{Reusing indexing work.}
Another line reduces how often the indexer and selector run.  PIVOT scans the prefix through a
proxy for nearby queries, then reuses its result or refines an individual Top-$K$ from its
candidates~\cite{pivot2026}.  IndexCache shares selected indices across layers that omit their
own indexers~\cite{indexcache2026}; LongCat combines cross-layer reuse with hierarchical
indexing and streaming-aware layouts~\cite{longcat2026}.  These methods exploit redundancy
outside one score row, while their owner layers, full-indexer layers, or per-query refinement
stages still require candidate selection.  Reducing per-selection cost is therefore
complementary to reducing the number of selections.

\paragraph{Alternative memory and execution substrates.}
Sparse selection also creates opportunities after indices have been produced.  HiSparse keeps
the complete KV history in host memory and resolves selected entries through a bounded GPU cache
while preserving the indexer's choices~\cite{hisparse2026}.  KARAT instead proposes a
programmable near-memory design for storing index keys and executing scoring, selection, and
gathering, evaluated in simulation~\cite{karat2026}.  HiSparse composes with a GPU selector,
whereas KARAT replaces the GPU-resident boundary.  Both address capacity and data placement
rather than materialized-row selection cost.

\section{Conclusion}
\label{sec:conclusion}

Long-context sparse attention creates an exact Top-$K$ regime in which $K$ remains fixed while
ragged score rows grow to hundreds of thousands of elements.  HPC-Ops exploits the separation
between locating a compact upper tail and resolving its exact rank boundary.  A regular partial
view proposes a row-local coarse boundary, a complete-row pass certifies that proposal and
materializes its admitted candidates, and FP32 refinement resolves only the remaining frontier.
Underfilled proposals are detected and recovered before output is committed, preserving exact
Top-$K$ semantics on every path.

HPC-Ops maps this design to persistent and KV-split sampled kernels, complemented by row-local
and cooperative direct-exact kernels for other shapes.  On indexer-score captures from
Hy4-Preview, this portfolio is $1.29$--$1.75\times$ faster than the best verified external exact
implementation across the operator matrix, with a $1.55\times$ geometric-mean speedup.  The gain
persists over framework-derived long-context traces, reaching $1.36\times$ and $1.48\times$.
Matched controls show that the improvement comes primarily from replacing complete-row boundary
localization with a sparse current-row proposal.

Taken together, these results show that sampling is most useful here as a proposal mechanism
rather than an approximation to Top-$K$.  The sampled boundary controls the amount of work
presented to exact refinement, while complete-row certification determines validity and
activates recovery when necessary.  This separation preserves exact Top-$K$ semantics while
trading common-path traffic against rare recovery work.

\end{document}